\documentclass[11pt]{article}
\usepackage{mainsty}
\usepackage[utf8]{inputenc}
\usepackage[normalem]{ulem}
\usepackage{enumerate}
\usepackage{rotating}
\usepackage{adjustbox}
\usepackage{booktabs}

\newcommand{\LTO}{{\scriptscriptstyle\mathrm{LTO}}}
\newcommand{\nLTO}{{\scriptscriptstyle\mathrm{naive-LTO}}}
\newcommand{\wLTO}{{\scriptscriptstyle\mathrm{w-LTO}}}
\newcommand{\rLTO}{{\scriptscriptstyle\mathrm{rs-LTO}}}
\newcommand{\pLTO}{{\scriptscriptstyle\mathrm{powered-LTO}}}
\newcommand{\LRO}{{\scriptscriptstyle\mathrm{LRO}}}
\newcommand{\rLRO}{{\scriptscriptstyle\mathrm{rs-LRO}}}

\newcommand{\epb}{\text{{\tiny exact-placebo}}}
\renewcommand{\pb}{\text{{\tiny app-placebo}}}

\title{Inference for Synthetic Controls via Refined Placebo Tests\thanks{First version: January 11, 2023. Authors are listed alphabetically. We thank Alberto Abadie, Rina Barber, Kirill Borusyak, John Cherian, Peng Ding, Kevin Guo, Michael Howes, Guido Imbens, David Ritzwoller, Brad Ross, Jann Spiess, Liyang Sun, Stefan Wager, and other participants of IMSI workshop on ``Permutation and Causal Inference" and CODE@MIT 2023 for many helpful comments and questions. L.L. is grateful for the support of National Science Foundation grant DMS-2338464. TS also acknowledges support from the NSF Graduate Research Fellowship Program under Grant DGE-1656518. }}
\author{Lihua Lei\footnote{ Graduate School of Business and Department of Statistics (by courtesy), Stanford University; Email: \href{mailto:lihualei@stanford.edu}{\texttt{lihualei@stanford.edu}}}\qquad  Timothy Sudijono\footnote{Department of Statistics, Stanford University; Email: \href{mailto:tsudijon@stanford.edu}{\texttt{tsudijon@stanford.edu}} }}
\date{\today}

\begin{document}
\maketitle

\begin{abstract}
The synthetic control method is often applied to problems with one treated unit and a small number of control units. A common inferential task in this setting is to test null hypotheses regarding the average treatment effect on the treated. Inference procedures that are justified asymptotically are often unsatisfactory due to (1) small sample sizes that render large-sample approximation fragile and (2) simplification of the estimation procedure that is implemented in practice. An alternative is permutation inference, which is related to a common diagnostic called the placebo test. It has provable Type-I error guarantees in finite samples without simplification of the method, when the treatment is uniformly assigned. Despite this robustness, the placebo test suffers from low resolution since the null distribution is constructed from only $N$ reference estimates, where $N$ is the sample size. This creates a barrier for statistical inference at a common level like $\alpha = 0.05$, especially when $N$ is small. We propose a novel leave-two-out procedure that bypasses this issue, providing $O(N^2)$ reference estimates while still maintaining the same finite-sample Type-I error guarantee under uniform assignment for a wide range of $N$ that is common in applications. Unlike the placebo test whose Type-I error always equals the theoretical upper bound, our procedure often achieves a lower unconditional Type-I error than theory suggests; this enables useful inference in the challenging regime when $\alpha < 1/N$.  Empirically, our procedure achieves a higher power when the effect size is reasonably large and a comparable power otherwise. We generalize our procedure to non-uniform assignments and show how to conduct sensitivity analysis. From a methodological perspective, our procedure can be viewed as a new type of randomization inference different from permutation or rank-based inference, which is particularly effective in small samples.
\end{abstract}

\section{Introduction}\label{sec:intro}

Synthetic control methods were first introduced by Abadie and Gardeazabal  \cite{abadie2003economic} to analyze the effects of terrorism in the Basque Country on the economy of the region. Since then, the method has developed into a powerful and popular tool in comparative case studies and program evaluation with very few treated units  \cite{abadie2010synthetic, cavallo2013catastrophic, billmeier2013assessing, abadie2015comparative, castillo2015tourism, gobillon2016regional, peri2019labor, donohue2019right, 
ben2021augmented, abadie2021using, sun2023using}. We focus on applications of the synthetic control framework in panel data settings with one treated unit. Even in applications with multiple treated units, this setting is still relevant if the researcher wants to assess effect heterogeneity. 

Suppose we observe data with $N$ units, $T$ time periods, where the $I$-th unit is treated from time $T_0+1,\dots,T$. The synthetic control (SC) method models the treated unit by a \textit{synthetic control}, a convex combination of the other control units which resembles  the treated unit as closely as possible according to pre-treatment outcomes and covariates. That is, we seek a vector of weights $W = (w_1,\dots,w_{I-1},w_{I+1},\dots,w_N)$ such that $w_{j} \geq 0$ for all $j$ and $\sum_{j \neq I}^{N} w_j = 1$, which corresponds to a weighted average of control units. The method chooses the weights $W$ by minimizing the norm
\begin{equation}\label{eq:SC_objective}
\norm{X_I - X_0 W}_V = \sqrt{(X_I - X_0W)^\top V (X_I - X_0W)}
\end{equation}
where $X_I \in \bR^p$ are covariates, or pre-treatment outcomes, for the treatment unit and $X_0 \in \bR^{p \times (N-1)}$ is a matrix whose rows are covariates for the control units. Finally, $V \in \bR^{p \times p}$ is a matrix which weights the different covariates, which may be either pre-specified or learned from the data \cite{abadie2011synth}. Importantly, in many flagship applications of the method, there is one treated unit and only a small number of control units. Moreover, the outcomes are typically measured over a short time horizon. For example in \cite{abadie2003economic}, $N = 17, T = 43, T_0 = 15$ and in \cite{abadie2010synthetic}, $N = 39, T = 31, T_0 = 19$. See Appendix \ref{sec:smalldata} for a  review of 31 papers that apply the synthetic control method, from which it is clear that applications to small datasets are very common.

By modelling the control counterfactuals for unit $I$, the SC method can be used to make inference on the sample average treatment effect on the treated. There are two popular approaches to inference for synthetic controls. The first imposes factor-model type structure on the outcomes, and does asymptotic inference. Many works in the literature, such as \cite{xu2017generalized,
amjad2018robust, athey2021matrix,
arkhangelsky2021synthetic,
ben2021augmented, 
cattaneo2021prediction, 
ferman2021synthetic, cattaneo2022uncertainty, ben2022synthetic, imbens2023identification, shen2023same, sun2023using, chernozhukov2018t, li2020statistical}, impose this assumption to study properties of the synthetic control estimator under the assumption that $N$ or $T$ is large. Most of these methods even require both the number of treated units and the number of post-treatment periods to be large. While various asymptotic approaches to inference exist, we focus on settings where small sample sizes do not justify the assumptions behind asymptotic inference. Furthermore, the asymptotic approaches often simplify the synthetic control procedures substantially to obtain provable guarantees. For example, they often do not consider a data-driven weight matrix $V$ that proves to be crucial in practice \cite{abadie2021using}.

The second approach is the design-based framework, where the data are seen as fixed, and the treatment unit $I$ is assumed to be uniformly assigned from $1$ to $N$. The placebo test of \cite{abadie2003economic, abadie2010synthetic} is an example; it was further extended in \cite{doudchenko2016balancing, firpo2018synthetic, bottmer2021design, shaikh2021randomization}. The uniform assignment assumption holds by design in experimental settings \cite{doudchenko2019designing, jones2019synthetic, abadie2021synthetic}. Even in observational settings, as argued by \cite{bottmer2021design}, the assumption is ``a natural starting point for many analyses, often after some adjustment for observed
covariates. In particular, many of the analyses of SC (synthetic control) methods explicitly or implicitly refer to units being comparable". The assumption could be more tenable after adjusting for covariates, as \cite{bottmer2021design} suggest. The assumption can also be justified in quasi-experimental settings \cite{arkhangelsky2021double, borusyak2020non, borusyak2024design}. Examples such as \cite{dupont2015happened, coffman2012hurricane}, where synthetic controls are used to analyze the impact of natural disasters, provide another setting for this assumption. The rate of natural disasters at a particular location can be estimated from past data, and is plausibly random over small geographic locations.  In this paper, we consider a design-based framework for inference.

The standard inference procedure in the design-based framework is the placebo test. It was introduced in \cite{abadie2003economic, abadie2010synthetic} as a diagnostic tool related to permutation inference and Fisher's randomization test. The procedure was formalized and expanded upon in \cite{firpo2018synthetic} to make inference under the assumption that $I$ is uniformly assigned. It creates a $p$-value via the following recipe. For every unit $i$, create its synthetic control using all other units and obtain predictions $\hat{Y}_{it}$ for all $t > T_0$. Then measure the goodness-of-fit by comparing $\hat{Y}_{it}$ with the observed outcomes through any summary statistic $R_i$, such as the widely used RMSPE statistic defined in \eqref{eq:RMSPE}. Finally, the exact placebo $p$-value is defined as 
\begin{equation}\label{eq:exact_placebo}
p_{\epb} = \frac{1}{N} \sum_{i=1}^N \mathbf{1}\set{R_I \leq R_i}.
\end{equation}
The standard argument for permutation tests shows that $p_{\epb}$ is a valid $p$-value.

\paragraph{Small sample issues of the placebo test.} The $p$-value obtained from such tests are coarse: they take values in the grid \sloppy $\set{\frac{1}{N}, \frac{2}{N},\dots,1}.$ When $N$ is relatively small like in the premier applications of the synthetic control method \cite{abadie2003economic,abadie2010synthetic, abadie2015comparative} or those examples in Appendix \ref{sec:smalldata}, the smallest value $1/N$ of the $p$-value may give statistical significance, while the $2/N$ may switch this conclusion. This low level of granularity hinders the application of the placebo test and related inference methods in practice. In some applications, the value of $1/N$ may be too large to even reject at a pre-specified level $\alpha$. For example, when $\alpha = 0.05$, $\alpha < 1/N$ in both \cite{abadie2003economic} and \cite{abadie2015comparative}. Due to a desire to control size at traditional levels like $\alpha \in \set{0.01,0.05,0.10}$ and generally small datasets, $N\alpha < 3$ is typical for a variety of applications. This includes state/province/prefecture-level analysis \cite{abadie2010synthetic, castillo2015tourism, dupont2015happened, donohue2019right, hankins2020finally,
ben2021augmented, ben2022synthetic}, district-level analysis within a state \cite{sun2023using}, and manufacturer-level analysis \cite{fremeth2013making}. Even in country-level or county/city-level analysis where the donor pool appears to be large, the number of comparable control units may be small due to restrictions based on geographical proximity \cite{coffman2012hurricane, billmeier2013assessing}, potential spillovers \cite{dupont2015happened, ando2015dreams}, economic prosperity \cite{abadie2015comparative}, market conditions \cite{peri2019labor,boes2012effect}, and eligibility \cite{newiak2017evaluating}. On the other hand, it is a common practice to rule out control units with poor pre-treatment fits for the placebo test \cite{abadie2010synthetic, almer2012effect}, which can result in a fair drop in the sample size.

More importantly, the placebo test is powerless in small data settings where $\alpha < 1/N$. For each unit $i$ and time $t$, let $(Y_{it}(1), Y_{it}(0))$ denote the pair of potential outcomes had the unit $i$ been treated and not treated at time $t$, respectively. Further, let $\cl{Y} = \set{(Y_{it}(1),Y_{it}(0)): i \in [N], t \in [T]}$ denote the realizations of all potential outcomes and assume they take distinct values almost surely. Then given $I$ is uniformly assigned, $p_{\epb}$ is distributed uniformly on the set $\set{1/N,2/N,\dots,1}$ under the sharp null $H_0: Y_{It}(1) = Y_{It}(0)$ for all $t$. Thus, the conditional  Type-I error is
\[
\Prob_{H_0}(p_{\epb} \leq \alpha | \cl{Y}) = \frac{\floor{ N\alpha}}{N}.
\]
Averaging over $\mathcal{Y}$, the unconditional Type-I error is 
\[  \Prob_{H_0}(p_{\epb} \leq \alpha)  = \frac{\floor{ N\alpha}}{N}.
\]
As a result, when $\alpha < 1/N$, the placebo test has zero size. Proposition \ref{prop:no_power} extends this argument to show that the placebo test has no power in this setting against a natural class of alternatives.

Can we develop a valid $p$-value which has power in these $\alpha < 1/N$ settings? One useful remedy is to define an approximate version of $p_{\epb}$ by not comparing with the treated unit itself:
\begin{equation}\label{eq:inexact_placebo}
p_{\pb} = \frac{1}{N} \sum_{i\neq I}\mathbb{I}\set{R_I\le R_i} = p_{\epb} - \frac{1}{N}.
\end{equation}
Clearly, its conditional  and unconditional Type-I errors are both $(\floor{N\alpha} + 1) / N$, which is strictly larger than $\alpha$. The textbook solution to resolve this dilemma (e.g., \cite{romano2005testing}) is to use a randomized $p$-value $p_{\epb} - U/N$ where $U$ is an independent draw from a uniform distribution on $[0, 1]$. While this randomized $p$-value has an exact Type-I error control and allows to reject the null even when $\alpha < 1/N$, it offers a cheap way of $p$-hacking -- the researcher can simply choose a random seed that generates a large $U$ without revealing its value. Taking this strategic behavior to the extreme, the randomized $p$-value is no different from the approximate placebo $p$-value defined in \eqref{eq:inexact_placebo}. Aside from the $p$-hacking incentive, the approximate $p$-value defined in \eqref{eq:inexact_placebo} rescaled by $N/(N-1)$ is sometimes used implicitly in the sense that significance is claimed when $p_{\pb} = 0$, i.e., the treated unit has the highest summary statistic. 

For the aforementioned reasons, the approximate placebo is unsatisfactory. Is there any alternative we may use? Unfortunately, no conditionally finite-sample-valid non-randomized $p$-value has power here, as Proposition \ref{prop:no_power} in Appendix \ref{subapp:sec1} shows. However, this does not imply unconditional impossibility. 
In applications of synthetic controls, it is reasonable to treat potential outcomes as random. Thus, it suffices to seek an unconditionally-valid p-value.
One way to get useful inference is to reject if the $p$-value $p(I) \leq 1/N$, and hope that $\Prob_{H_0}(p(I) \leq 1/N | \cl{Y}) = 0$ for some $\cl{Y}$. Averaged over different datasets $\cl{Y}$, we may hope for an unconditionally valid $p$-value. For example, when $1/2N < \alpha < 1/N$, if $\Prob_{H_0}(p(I) \leq \alpha | \cl{Y})$ takes values $0$ and $1/N$ both with probability $1/2$, then the unconditional Type-I error $\Prob_{H_0}(p(I) \leq \alpha) = 1/2N < \alpha$. Neither the exact nor the approximate placebo test has this property because their conditional Type-I errors are both constant regardless of $\mathcal{Y}$. Thus, we cannot achieve this goal by modifying the standard permutation test.

\subsection{Contributions} We present the Leave-Two-Out (LTO) placebo test, which gives up conditional validity for a shot at unconditional validity. The LTO placebo test is a non-randomized approximate $p$-value which can be seen as a leave-two-out extension of the usual placebo test, with a user-specified statistic. The LTO placebo test has several advantages over the approximate placebo inference method. 
\begin{itemize}
    \item \underline{Non-randomized inference when $N\alpha$ is small.} The LTO placebo test has the same \textit{conditional} Type-I error guarantee as the approximate placebo test \eqref{eq:inexact_placebo} for a wide range of $N,\alpha$ seen in practice. It is true in particular when $N\alpha \le 1$ or $N\alpha < 3.9$ if further $N > 10$. Note that the latter case covers all examples listed above. This is proved in Theorem \ref{cor:TypeIguarantee_alphaless1/n}. In addition, we provide empirical evidence that the LTO placebo test often achieves a much lower Type-I error than our theory indicates. As a result, unlike the approximate placebo test \eqref{eq:inexact_placebo}, the LTO placebo test can be \emph{unconditionally valid}, and often matches the unconditional Type I error of the exact placebo test.
    \item \underline{Power Improvements.} We empirically observe power gain compared to 
 the approximate placebo test when the effect size is large, and comparable power for intermediate effect sizes. We also observe strict power gain for all effect sizes compared to the exact placebo test and randomized placebo test even when $\alpha > 1/N$. 
 Theoretically, we extend the classical notion of consistent tests to the fixed-$N$ setting, and prove that the LTO placebo test is uniformly consistent while the approximate placebo test fails to be when $\alpha < 1/N$.
    \item \underline{Granularity Improvements.} By leaving multiple data points out at a time in the same procedure, the leave-two-out $p$-value lies on a much more refined scale than $\set{\frac{1}{N}, \frac{2}{N},\dots,\frac{N-1}{N},1},$ solving granularity problems of the placebo test. 
    
\end{itemize}

In addition, we discuss a weighted variant of the LTO placebo test under non-uniform assignment probabilities, a sensitivity analysis for the procedure, other versions of LTO placebo tests, and leave-$r$-out extensions for $r > 2$.

\subsection{Example: California Proposition 99 Dataset}
\label{sec:intro_example}
In this section, we consider a semisynthetic simulation using the California proposition 99 example of \cite{abadie2010synthetic}, one of the flagship applications of the synthetic control method. The dataset records cigarette sales, prices, and various other regional attributes in 39 US states, over the years 1970 to 2000. In \cite{abadie2010synthetic}, the data are used in the synthetic control procedure to analyze the effect of Proposition 99, a 1988 California state law increasing cigarette taxes, on smoking. In this example, $N = 39, T = 30, T_0 = 19.$

We construct a semisynthetic version of the data, where California is removed from the dataset. A subsample of $N = 30$ units is randomly chosen from the remaining set of $38$ units. In the smaller dataset, a state is chosen uniformly at random to be the `treated' unit. For the treated unit $I$, we posit a uniform treatment effect $\tau$, and we update the observed outcomes $Y_{I,t}$ by adding $\tau$. A range of values of $\tau$ are tested, corresponding roughly to multiples of $(0,-1,-2,-3)$ of the standard deviation of the outcome variable (cigarette sales). When $\tau = 0$ we are calculating the Type-I error of the procedures. Further 
detail on the setup of this simulation can be found in Section \ref{sec:powersimulations}. 

In the first setting, we take $\alpha = 0.02$ $(\alpha < 1/N)$. We run the approximate placebo test, its randomized version, and the LTO placebo test in order to compare their Type-I error and power. See Section \ref{sec:powersimulations} for more detail. The result of this comparison are shown in Figure \ref{fig:smoking_intro_powercomparison}. Several observations can be made regarding empirical performance. Firstly, the Type-I errors of the LTO $p$-values are on average smaller than $1/N$, and is below $\alpha$ in each simulation. This means that for a majority of the Monte Carlo resamples of the dataset, the empirical Type-I error was zero. Thus, the LTO procedure is unconditionally valid. On the other hand, the Type-I error of the approximate placebo test is always fixed. We do not present the exact placebo test because it is powerless. Simultaneously, the power of the LTO placebo test is better than that of the approximate placebo test when the effect size is large. When the effect size is intermediate, however, the approximate placebo test has higher power, but not by much. In Section \ref{sec:boosting_power}, we propose a more powerful test, which has nearly identical power as the approximate placebo test in this setting, by modifying the LTO p-value for any given $\alpha$; see Figure \ref{fig:smoking_full_powercomparison}. Lastly, the randomized placebo test controls Type-I error at the desired level, but has much smaller power than the other two methods. Further it is susceptible to adverse incentives to $p$-hack.

\begin{figure}
  \begin{subfigure}{0.5\textwidth}
    \centering
    \includegraphics[width = \linewidth]
    {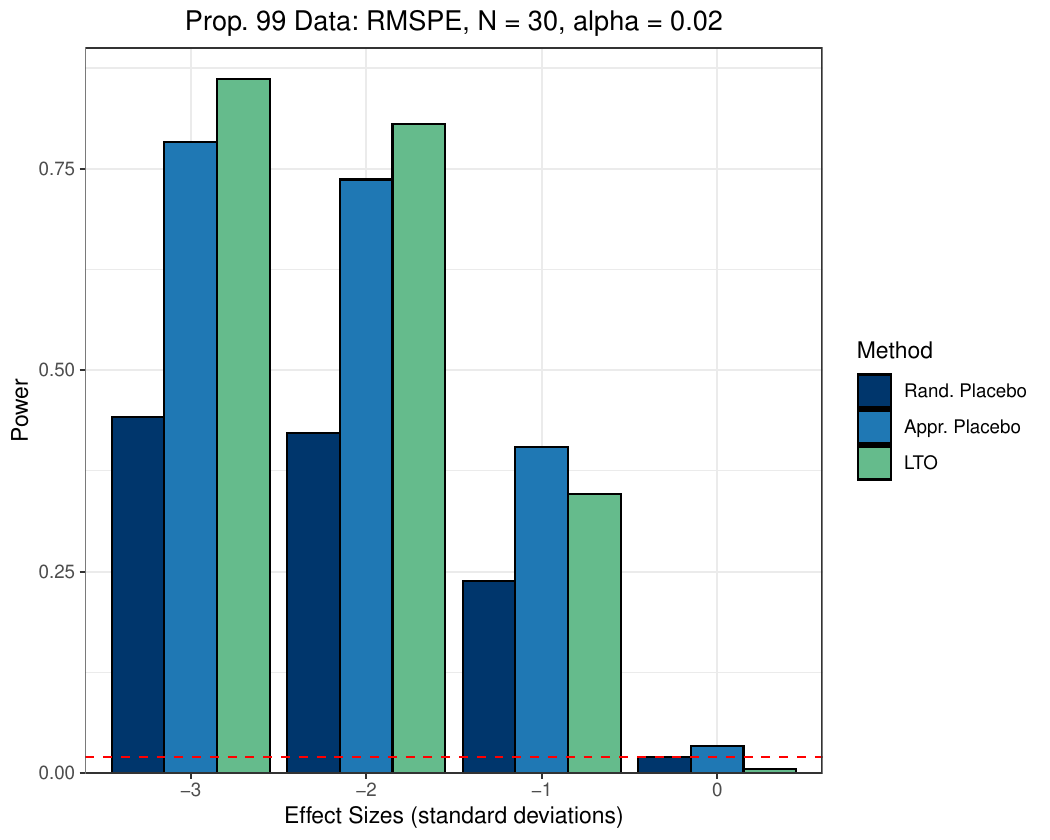}
    \label{fig:smoking_rmspe_n30_alpha0.02_intro}
  \end{subfigure}%
  \begin{subfigure}{0.5\textwidth}
    \centering
    \includegraphics[width=\linewidth]{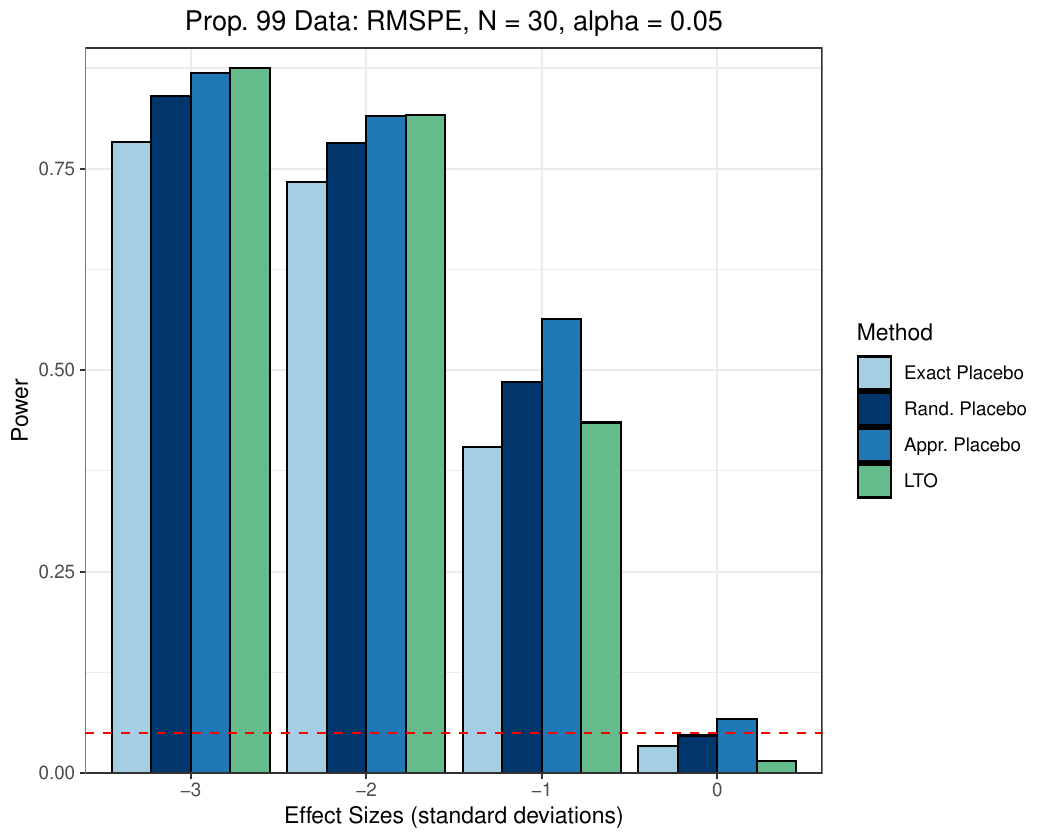}
    \label{fig:smoking_rmspe_n30_alpha0.05_intro}
  \end{subfigure}
  \caption{Power of several different $p$-values versus effect size, on twenty subsamples of the California Proposition 99 dataset, of size 30. The $x$-axis is the effect size $\tau$ scaled in terms of multiples of the standard deviation of the outcome variable of the dataset. Results are averaged over 20 Monte Carlo runs. The exact placebo, randomized placebo, approximate placebo, and LTO placebo are compared, all constructed using the RMSPE statistic. The rightmost column with $\tau = 0$ is just the Type-I error of the procedures. Red dashed line indicates the level $\alpha$. (Left) $\alpha = 0.02$, which falls in the $\alpha < 1/N$ regime. The exact placebo is omitted because the power is zero in all cases. (Right) $\alpha = 0.05$, which falls in the $\alpha > 1/N$ regime.}
  \label{fig:smoking_intro_powercomparison}
\end{figure}

In the second setting we take $\alpha = 0.05 \,\, (\alpha \geq 1/N)$. We compare the exact placebo test, the randomized placebo test, the approximate placebo test, and the LTO placebo test; see the right side of Figure \ref{fig:smoking_intro_powercomparison}.
We can largely draw the same conclusions as with the $\alpha < 1/N$ setting. While the power gap between the LTO placebo test and approximate placebo test is larger for intermediate effect sizes, a more powerful LTO placebo test developed in Section \ref{sec:boosting_power} reduces about half of the gap; see Figure \ref{fig:smoking_full_powercomparison} for a full comparison. We again observe that the LTO procedure is unconditionally valid. In most Monte Carlo resamples, the Type-I error is in fact zero. Notably, the unconditional Type-I error of the LTO placebo test is even lower than that of the exact placebo test, which is below $\alpha$. Moreover, the power of the LTO placebo test is consistently larger than that of the exact placebo test. In this case, the randomized placebo test interpolates between the exact and approximate placebo tests and controls Type-I error at the desired level. See Section \ref{sec:simulations} for a detailed description of the numerical experiment, and results for similar datasets.

\paragraph{Real Data Analysis.} We replicate the analysis of California Proposition 99's effect on smoking using both the placebo $p$-value and our LTO method. Here, $N = 39$ with one treated unit. To highlight a sensitivity analysis procedure we develop for the LTO $p$-value, we conduct inference with $\alpha = 0.05$. The LTO $p$-value obtained was $0.024$ while the exact placebo $p$-value was $0.026$, which exactly equals $1/N$. The approximate placebo $p$-value was $0$. 

In this setting, we can also inspect how sensitive the conclusion of non-significance is with respect to the equal weights assumption. We apply a Rubin-Rosenbaum style sensitivity analysis procedure to this example, in order to assess the sensivity of inference to the equal weights assumption; see Section \ref{sec:sa} for more details. The sensitivity analysis relies on a weighted version of the LTO $p$-value, which is also an approximate $p$-value under non-uniform treatment probabilities $\pi_i$ for unit $i$. Fixing a $\Gamma,$ we constrain the non-uniform propensities $\pi_i$ to sum to $1$ and satisfy $\pi_i \in [\frac{1}{\Gamma N}, \frac{\Gamma}{N}]$ (similarly to \cite{rosenbaum1983assessing}), and calculate the largest possible value of the weighted $p$-value for propensities in this constraint. Finally, we search for the smallest value of $\Gamma$ such that the conclusion of significance is overturned. 

As the reported LTO $p$-value is $0.024$, we want to see how large $\Gamma$ can be before the conclusion of significance is possibly overturned.  The output is shown in Figure \ref{fig:prop99_sa}. From the figure, $\Gamma$ is around $1.4$ before the maximum possible $p$-value is greater than $0.05$. This suggests the conclusion of significance is robust under moderate confounding.

\begin{figure}
    \centering
    \includegraphics[width = 0.5\textwidth]{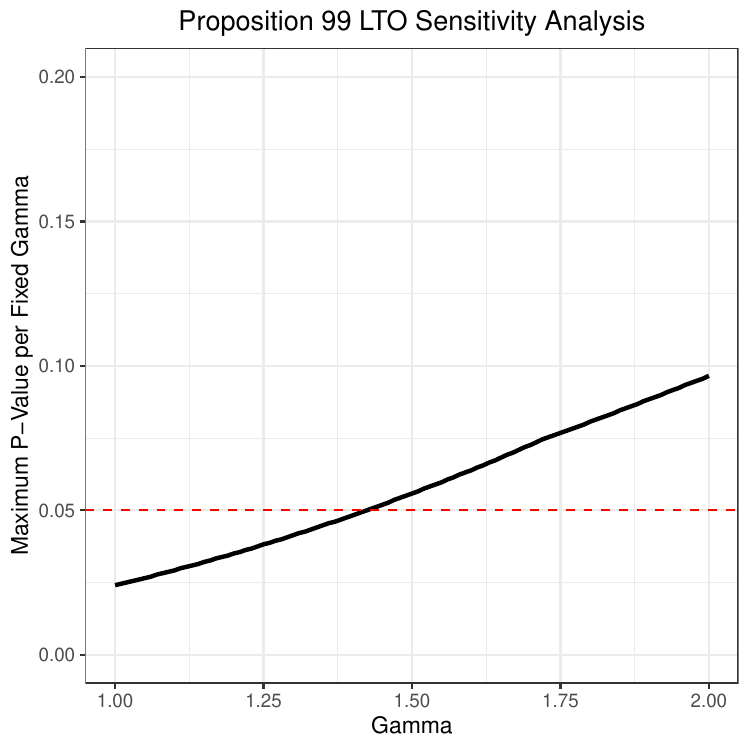}
    \caption{Output of a sensitivity analysis for the Proposition 99 smoking dataset of \cite{abadie2011synth}. The red dashed line signifies the level $\alpha = 0.05$. The RMSPE statistic was used in the construction of the synthetic control.}
    \label{fig:prop99_sa}
\end{figure}

\section{Leave-Two-Out Placebo Test}
\label{sec:ltojk}

Consider a design-based setup similar to that of the placebo test. We observe a panel $Y_{it}$ of $N$ units, $T$ time periods, where the $I$th unit is treated from time $T_0+1,\dots,T$. In potential outcomes, the data is given as:
      \[
    \begin{bmatrix}
    Y_{1,1}(0) & Y_{1,2}(0) & \dots & Y_{1,T_0 + 1}(0)&  \dots & Y_{1,T-1}(0) & Y_{1,T}(0) \\
    Y_{2,1}(0) & Y_{2,2}(0) & \dots & Y_{2,T_0 + 1}(0) & \dots & Y_{2,T-1}(0) & Y_{2,T}(0) \\
    \vdots & \vdots & \ddots & \vdots & \ddots & \vdots & \vdots \\
    Y_{I,1}(0) & Y_{I,2}(0) & \dots & \textcolor{blue}{Y_{I,T_0+ 1}(1)} & \textcolor{blue}{\dots} & \textcolor{blue}{Y_{I,T-1}(1)} & \textcolor{blue}{Y_{I,T}(1)} \\
    \vdots & \vdots & \ddots & \vdots & \ddots & \vdots & \vdots \\
    Y_{N,1}(0) & Y_{N,2}(0) & \dots & Y_{N,T_0+ 1}(0) & \dots & Y_{N,T-1}(0) & Y_{N,T}(0)
    \end{bmatrix}.
    \]
We assume that the treatment index $I$ is uniformly assigned across the $N$ units. Synthetic control methods estimate the average treatment effect on the treated, $\tau_{I,t} = Y_{I,t}(1) - Y_{I,t}(0)$, by using the synthetic control $\hat{Y}_{I,t}$ to estimate the control counterfactual. Following the literature \cite{abadie2010synthetic, abadie2015comparative, firpo2018synthetic}, the primary inference task of interest is to test the sharp null hypothesis 
\begin{equation}\label{eq:global_null}
H_0: Y_{i,t}(1) = Y_{i,t}(0), \quad \forall i, \forall t \geq T_0 + 1
\end{equation}
Extensions to sharp nulls of the form $H_0: Y_{i,t}(1) = Y_{i,t}(0) + \tau_{i,t}$ for known quantities $\tau_{i,t}, i =1,\dots,N,$ and $t \geq T_0 + 1$ can also be handled.

\subsection{Naive LTO Placebo Test}
\label{sec:naive_lto}

Our LTO placebo test procedure works as follows: 
\begin{enumerate}
    \item For every pair of distinct units $(i,j)$ from $1,\dots,N$ and unequal to $I$, create the synthetic control with control units $1,\dots,N$ excluding points $i,j,I$ to obtain $\hat{\mathbf{Y}}_{k}^{-(i,j,I)}, k \in \set{i,j,I},$ the vector of synthetic control outcomes for unit $k$ using controls $[N]\setminus \set{i,j,I}.$ For some statistic $S(\cdot,\cdot),$ define the residual
    \begin{equation}\label{eq:general_score}
    R_{i,j,I;k} = |S(\mathbf{Y}_{k},\hat{\mathbf{Y}}^{-(i,j,I)}_{k})|,
    \end{equation}
    where $k \in \set{i,j,I}$ and $\mathbf{Y}_k$ is the observed data for unit $k.$
    
     A common choice is the post-period to pre-period Ratio of Mean Square Prediction Error (RMSPE) statistic proposed in \cite{abadie2015comparative}: 
    \begin{equation}\label{eq:RMSPE}
    S_{{\tiny \text{ratio-RMSPE}}}(\mathbf{Y}_k,\hat{\mathbf{Y}}_k) =
    \frac{S_{{\tiny \text{post-RMSPE}}}(\mathbf{Y}_k,\hat{\mathbf{Y}}_k)}{S_{{\tiny \text{pre-RMSPE}}}(\mathbf{Y}_k,\hat{\mathbf{Y}}_k)} = \frac{\sum_{t = T_0 + 1}^T (Y_{k, t} - \hat{Y}_{k, t})^2 / (T - T_0)}{\sum_{t = 1}^{T_0} (Y_{k, t} - \hat{Y}_{k, t})^2 / T_0}.
    \end{equation}
    
    \item For every pair of indices in the first step, compute the Leave-Two-Out residual
    \[
    R_{i,j,I}^{\LTO} := \max(R_{i,j,I;i}, R_{i,j,I;j}).
    \]
    Denote the event that $I$'s residual is larger than that of $i,j$ by $\mathbb{I}\set{I \succ i,j}$, so that
    \[
    \mathbb{I}\set{I \succ i,j} := \mathbf{1}\set{R_{i,j,I;I} > R_{i,j}^{\LTO}}.
    \]
    \item Compute the quantity
    \begin{equation}
    \label{eq:jkplus_lto_pvalue}
    p_{\nLTO} := \frac{1}{(N-1)(N-2)} \sum_{\substack{i,j \in[N] \setminus I \\ i \neq j}} \mathbb{I}\set{I \not\succ i,j},
    \end{equation}
    which functions as an approximate $p$-value under the null hypothesis. 
\end{enumerate}

\begin{remark}\label{rem:tournament}
Thinking of the LTO placebo test $p$-value in terms of tournaments is useful. Consider a tournament between all of the $N$ units, where matches are held for every triple of units. A unit $i$ wins a match amongst units $i,j,k$ if its residual $R_{i,j,k;i}$ is the largest. Then $p_{\nLTO}$ after renormalization counts the number of matches in which unit $I$ did not win.
\end{remark}




By approximate $p$-value, we mean that the LTO $p$-value satisfies the following Type-I error guarantee.
\begin{theorem}
\label{cor:TypeIguarantee_alphaless1/n} Let $N \ge 3$. Then the following results hold. 
\begin{enumerate}[(a)]
\item When $\alpha \le \frac{1}{N-1}, \Prob_{H_0}(p_{\nLTO} \leq \alpha) \leq \frac{1}{N}$. 
\item When $N > 6$, 
\[\Prob_{H_0}(p_{\nLTO} \leq \alpha) \leq \frac{\floor{N\alpha}+1}{N}, \quad \text{for any }N\alpha < 4\left(1 - \frac{2}{(N - 1)(N- 
2)}\right).\]
\item For any integer $k$ such that $\frac{k}{N-1}\left(\frac{N}{N-2}\frac{k-1}{3} - 1\right) < 1$, 
\[
\Prob_{H_0}(p_{\nLTO} \leq \alpha) \leq \frac{k}{N}, \quad \text{for any }N\alpha \in \left(k-1, k - \frac{k}{N-1}\left[\frac{N}{N-2}\frac{k-1}{3} - 1\right]\right) \text{and }\alpha <2/3.
\]
\end{enumerate}
\end{theorem}
The proof of Theorem \ref{cor:TypeIguarantee_alphaless1/n} is presented in Appendix \ref{subapp:sec2}. It is a (non-trivial) consequence of the following more general result. We sketch its proof at the end of the section. 
\begin{theorem}
\label{thm:TypeIerror}
Under the assumption of uniform treatment, for any $\alpha < 2/3$, 
\begin{equation}
\Prob_{H_0}(p_{\nLTO} \leq \alpha) \leq   \frac{\floor{N f(N,\alpha)}}{N},
\end{equation}   
where 
\begin{equation}\label{eq:finite_samples_TypeI_Error_bound}
f(N, \alpha) :=\frac{3 - \frac{3}{N} - \sqrt{9(1-\frac{1}{N})^2 - 12\left(-\frac{4}{3N^2} + \frac{1}{N} + \alpha(1 - \frac{1}{N})(1 - \frac{2}{N})\right) }}{2}.
\end{equation}
\end{theorem}

Theorem \ref{cor:TypeIguarantee_alphaless1/n} guarantees that $p_{\nLTO}$ never has worse conditional Type I error than $p_{\pb}$, for a wide range of parameters. Specifically, by Theorem \ref{cor:TypeIguarantee_alphaless1/n} (b), when $\alpha = 0.05$, it implies that the LTO $p$-value has the same theoretical guarantee with the placebo $p$-value for any $N \le 79$. This covers a wide range of settings where the synthetic control method is applied; see Appendix \ref{sec:smalldata}. By Theorem \ref{cor:TypeIguarantee_alphaless1/n} (c), we can easily check if the naive LTO placebo test has the same Type-I error guarantee as the approximate placebo test for any $N$ and $\alpha$. In particular, for all $6 < N < 200$, this equivalence holds for $\alpha \in \{0.01, 0.02\}$. For $\alpha = 0.05$, it holds except for $N \in \{80, 100, 120, 139, 140, 159, 160, 179, 180, 198, 199\}$. To further gain some insight for larger $\alpha$, we plot the Type-I error bound in Theorem \ref{thm:TypeIerror} for $N = 17$ in Figure \ref{fig:TypeIerrorbounds}. The LTO p-value offers guarantees that are comparable to or slightly stronger than those of the approximate placebo test for all practical levels, i.e., when $\alpha \le 0.25$. Lastly, when $N$ is large, the bound is close to $\frac{3 - \sqrt{9 - 12\alpha}}{2}$, which is very close to $\alpha$ when $\alpha$ is small. For $\alpha = 0.05$, the right hand side is around $0.0508$. For $\alpha = 0.1,$ the bound comes out to $0.104.$

On the other hand, as suggested by our empirical results in Section \ref{sec:simulations}, the actual Type-I error of the LTO $p$-value can be strictly lower than the bound while the bound for the placebo $p$-value is tight. Combined with the improved granularity and improved power documented in Sections \ref{sec:powersimulations} and \ref{sec:consistency}, $p_{\LTO}$ provides a compelling alternative to $p_{\pb}$ for randomization inference.

\begin{figure}[h]
    \centering
    \includegraphics[width = \textwidth]{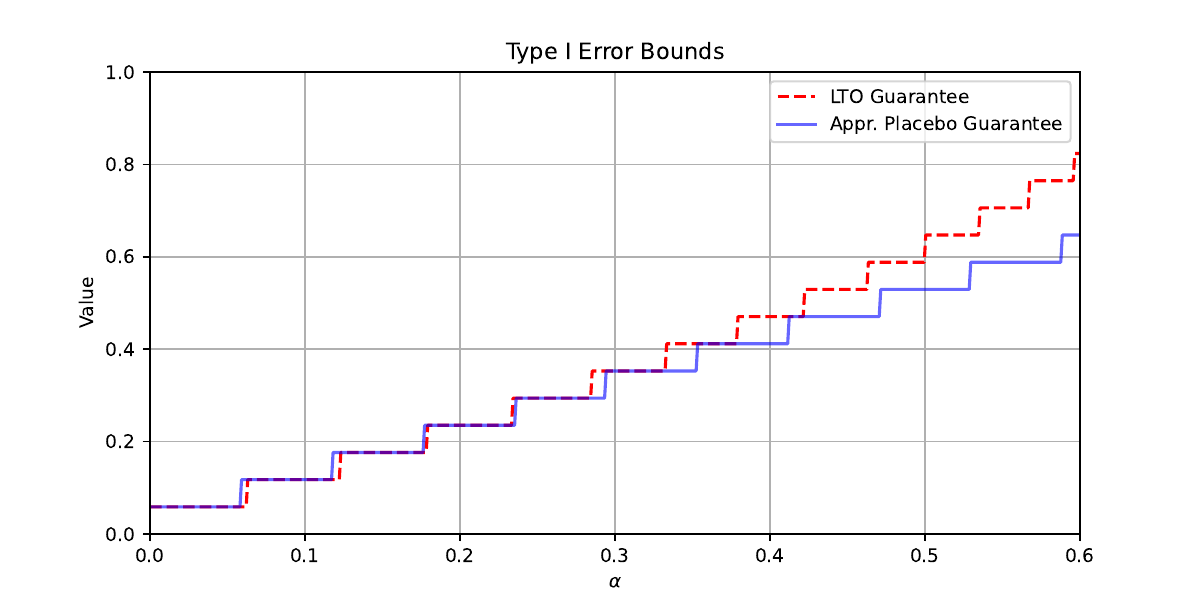}
    \caption{ The red dashed line shows the Type-I error bounds $\floor{\frac{N f(N,\alpha)}{N}}$ in Theorem \ref{thm:TypeIerror} of the Naive LTO procedure, which is a non-decreasing piecewise-constant function in $\alpha$. Here, $N = 17$. The stepwise constant solid blue line is the approximate placebo guarantee, $\frac{\floor{N\alpha}+1}{N}$.
    }
    \label{fig:TypeIerrorbounds}
\end{figure}



\paragraph{Proof Sketch of Theorem \ref{thm:TypeIerror}.} The proof of the coverage guarantee relies on analyzing combinatorial structures related to tournaments. We present the proof here in order to introduce the main ideas of the analysis; it is an interesting extension of the proof of coverage in the Jackknife+ \cite{barber2021predictive}. Let us recall the notation
\[
\bb{I}(k \succ i,j) := \mathbf{1}\set{R_{i,j,k;k} > \max(R_{i,j,k;i}, R_{i,j,k;j})}.
\]
Call a unit $k$ \textit{strange} if 
\[
\sum_{\substack{i,j \in[N] \setminus k \\ i \neq j}} \bb{I}(k \succ i,j) \geq (1-\alpha) (N-1)(N-2).
\]
Intuitively, the LTO $p$-value can be seen as the normalized rank of the unit $I$ in a tournament where matches are held between every triple of distinct units $i,j,k$. The winner of the match is the unit with the largest residual. We call a unit \textit{strange} if it wins too many matches, quantified above. In such a tournament, not everyone can win too often; following this intuition, we upper bound the number of strange units. \\

Let $\cl{S}$ denote the set of strange units, and let $s := |\cl{S}|$. Summing the definition of strangeness on both sides, we see that 
\[
\sum_{k \in \cl{S}} \sum_{i,j \in [N], i\neq j} \bb{I}(k \succ i,j) \geq (1-\alpha) s (N-1)(N-2).
\]
The sum on the left hand side may be split into three terms:
\begin{align*}
     & \overbrace{\sum_{i,j,k \in \cl{S}, i\neq j\neq k}  \bb{I}(k \succ i,j)}^{\text{(I)}}
    +  \overbrace{2\sum_{i,k \in \cl{S}, i\neq k, j \in \cl{S}^c} \bb{I}(k \succ i,j)}^{\text{(II)}}
    +  \overbrace{\sum_{ i,j \in \cl{S}^c, i\neq j, k \in \cl{S}}  \bb{I}(k \succ i,j)}^{\text{(III)}},
\end{align*}
where $i \neq j\neq k$ is shorthand notation for mutually distinct $i, j, k$.
For the first sum, note by renaming the labels that 
$\sum_{i,j,k \in \cl{S}, i\neq j\neq k} \bb{I}(k \succ i,j) = \sum_{i,j,k \in \cl{S}, i\neq j\neq k}  \bb{I}(i \succ j,k) = \sum_{i,j,k \in \cl{S}, i\neq j\neq k}  \bb{I}(j \succ k,i).$ Thus 
\[
\text{(I)} \leq \sum_{i,j,k \in \cl{S}, i\neq j\neq k} \frac{1}{3}\left(\bb{I}(k \succ i,j) + \bb{I}(i \succ j,k) + \bb{I}(j \succ k,i) \right) \leq s(s-1)(s-2)/3.
\]
Crucially, we use the fact that there can be at most one winner in a triple comparison: $\bb{I}(k \succ i,j) + \bb{I}(i \succ j,k) + \bb{I}(j \succ k,i) \leq 1$.\footnote{We allow for zero winners in the triple comparison if there are ties in the residuals.} For the second sum, swap the naming of the $i,k$ labels and use the bound $\bb{I}(k \succ i,k) + \bb{I}(i \succ k,j) \leq 1$ to see that 
\[
\text{(II)} \leq 2\sum_{i,k \in \cl{S}, i\neq k, j \in \cl{S}^c} \frac{1}{2}\left(\bb{I}(k \succ i,k) + \bb{I}(i \succ k,j) \right) \leq s(s-1)(N - s).
\]
Finally, in the last sum, we use the naive bound $\text{(III)} \leq s (N-s)(N-1-s).$ 

Combining these bounds and working through algebra, we arrive at a quadratic inequality in the variable $\beta := s/N$: 
\[
\frac{\beta^2}{3} - \beta(1 - \frac{1}{N}) - \frac{4}{3N^2} + \frac{1}{N} + \alpha(1 - \frac{1}{N})(1-  \frac{2}{N}) \geq 0.
\]
We can solve this by quadratic formula, yielding that, for any $\alpha < 2/3$, 

\begin{equation}
\label{eq:quadraticbound_jkproof}
\beta \leq f(N, \alpha).
\end{equation}

To conclude, under uniform treatment we have $\Prob(p_{\nLTO} \leq \alpha) = \Prob(I \in \cl{S}) = \beta.$ The bound in equation \eqref{eq:quadraticbound_jkproof} shows that $p_{\nLTO}$ is an approximate $p$-value. A full proof with further computational details may be found in Appendix \ref{sec:proofs}.

\subsection{Improved LTO Placebo test}
\label{sec:boosting_power}

Notice that in defining the notion of strange unit, we have great flexibility in choosing the threshold. Call a unit $k$ $c$-\textit{strange} if 
\[
\sum_{\substack{i,j \in[N] \setminus k \\ i \neq j}} \bb{I}(k \succ i,j) \geq (1-\alpha - c) (N-1)(N-2).
\]
This results in a new $p$-value $p_{\nLTO} - c$. By Theorem \ref{thm:TypeIerror}, 
\[\Prob(p_{\nLTO} - c \leq \alpha) = \Prob(p_{\nLTO}\leq \alpha + c)\leq \frac{\floor{N f(N,\alpha+c)}}{N}.\]
As a result, for any $c > 0$ such that
\[\frac{\floor{N f(N,\alpha+c)}}{N} = \frac{\floor{N f(N,\alpha)}}{N},\]
the adjusted $p$-value $p_{\nLTO} - c$ has the same bound on the rejection probability as the naive LTO $p$-value $p_{\nLTO}$ but is strictly more powerful than the latter, since it is strictly smaller. This motivates the improved LTO placebo $p$-value.

\begin{theorem}\label{thm:improved_LTO}
Let $\delta$ be any positive number (e.g., $\delta = 10^{-10}$). Define 
\begin{equation}\label{eq:improved_LTO_pvalue}
p_{\pLTO}(\alpha) = p_{\nLTO} - c(N, \alpha) + \delta,
\end{equation}
where
\[c(N, \alpha) = \min\{c: f(N,\alpha+c) = (\floor{Nf(N, \alpha)} + 1)/N\}.\]
Then 
\[\Prob(p_{\pLTO}(\alpha)\le \alpha) \le \frac{\floor{N f(N,\alpha)}}{N}.\]
\end{theorem}
The power improvement is enabled by the discrete Type-I error bound. Since the improved LTO placebo test has the same Type-I error bound as the naive one, Theorem \ref{cor:TypeIguarantee_alphaless1/n} applies to $p_{\pLTO}(\alpha)$. As a result, $p_{\pLTO}(\alpha)$ has the same theoretical guarantee as $p_{\pb}$ when $N\alpha < 4(1 - 2/(N-1)(N-2))$. Note that $p_{\pLTO}(\alpha)$ depends on $\alpha$ and hence, unlike $p_{\nLTO}$, cannot be viewed as a $p$-value. It can only be used for testing $H_0$ at level $\alpha$.

While the expression of $c(N, \alpha)$ is complicated, we can easily compute it by a binary search. For the Proposition 99 dataset, $N = 39$, $c(N, 0.02) = 0.006$ and $c(N, 0.05) = 0.002$. For the application in \cite{abadie2003economic}, $N = 17$ and $c(N, 0.05) = 0.0125$.

\paragraph{Further power improvement by enforcing decision monotonicity.}

One undesirable feature of $p_{\pLTO}(\alpha)$ is that the procedure is not decision-monotonic in $\alpha.$  That is, if the test based on $p_{\pLTO}(\alpha)$ rejects at some level $\alpha$, the test using $p_{\pLTO}(\alpha')$ may not reject at an $\alpha'$ with $\alpha' > \alpha$. This is due to the step-like behavior of the Type-I error upper bound.

We can enforce decision monotonicity and also improve power of the procedure by considering instead 
\[
I_\alpha = \max_{\alpha' \leq \alpha} \mathbf{1}\set{p_{\pLTO}(\alpha') \leq \alpha'},
\]
rejecting at level $\alpha$ if $I_\alpha = 1$. Clearly this procedure is more powerful than rejecting only when $p_{\pLTO}(\alpha) \leq \alpha$. The procedure also controls Type-I error. To see this, 
\begin{align*}
    \Prob(I_\alpha = 1) & = \Prob(\exists \alpha': p_{\pLTO}(\alpha') \leq \alpha')  = \Prob(\exists \alpha': p_{\nLTO} \leq \alpha' + c(N, \alpha') - \delta),
\end{align*}
The last line may be written
\begin{align*}
     & \Prob\left(p_{\nLTO} \leq \max_{\alpha' \leq \alpha} \alpha' + c(N, \alpha') - \delta\right) =  \Prob\left(p_{\nLTO} \leq  g(\alpha) + c(N, g(\alpha)) - \delta\right) ,
\end{align*}
where $g(\alpha) = \operatorname{argmax}_{\alpha' \leq \alpha} \alpha' + c(N, \alpha') - \delta.$ The latter is equal to $\Prob(p_{\pLTO}(g(\alpha)) \leq g(\alpha)) \leq \frac{\floor{N f(N,g(\alpha))}}{N} \leq \frac{\floor{N f(N,\alpha)}}{N}$, which is the same Type-I error bound as using just $p_{\pLTO}(\alpha).$ We do not use this procedure in our simulations, leaving this method for theoretical interest.

\subsection{Related Work}

The existing literature on synthetic controls is vast: see the survey article \cite{abadie2021using} for a review. Accordingly, there has been much work addressing inference for synthetic controls. Among the first inferential procedures for synthetic controls was the placebo test, first described in \cite{abadie2003economic}, and formalized by Firpo and Possebom \cite{firpo2018synthetic} into a permutation test. They work in the same design-based framework where the probability of treatment is uniform. Several extensions are made, including to a non-equally weighted inference procedure which is useful for sensitivity analysis. 

Other works have also considered the design-based setup and randomization or permutation-based inference. See \cite{bottmer2021design} for a design-based analysis of the synthetic controls procedure under a similar uniform treatment assumption, and also a random treatment time assumption. \cite{shaikh2021randomization} gives a related randomization inference procedure in a staggered adoption setting where there are multiple treated units, under the additional assumption that the time at which each unit adopts the treatment that follows a Cox proportional hazards model. In general, design-based inference proves to be a powerful tool for panel data analysis \cite{arkhangelsky2021double,rambachan2020design,
abadie2020sampling, athey2022design, roth2023efficient,xu2023difference}.

Inference procedures have also been developed assuming underlying factor model structure and related time series assumptions on the data. \cite{chernozhukov2021exact} develop an inference procedure assuming a regularity condition on the residual process in time. The paper contains an extensive discussion on the various choices for creating the counterfactual or synthetic control. Another line of research \cite{cattaneo2021prediction, cattaneo2022uncertainty} provides prediction intervals via non-asymptotic concentration bounds, under stationarity and weak dependence time series assumptions on the data. See also \cite{hahn2017synthetic, amjad2018robust, arkhangelsky2021synthetic, arkhangelsky2023large, li2020statistical, chernozhukov2018t} for related asymptotic analyses of the synthetic-control type estimators under time series assumptions. 

For many of these results, consistency of the estimator and the resulting inference procedure is shown in the limit as the number of controls and units goes to infinity. Although many of these works offer finite sample bounds, it is difficult to contextualize these results when the applications of synthetic controls are usually small $N$, small $T$ settings. Furthermore, many of these asymptotic results require simplifications on the synthetic control procedure to carry out analysis. The LTO placebo test can use other counterfactual models beyond the standard synthetic control method, and we require no conceptual simplifications to do inference. As a result, the LTO procedure may be used with generalizations of the original synthetic control procedure such as those outlined in \cite{doudchenko2016balancing, amjad2018robust,  chernozhukov2021exact, arkhangelsky2021synthetic, gunsilius2023distributional}. Furthermore, it can be applied with non-synthetic-control methods developed for inference with few clusters; see Section \ref{sec:few_clusters} for further discussion.

\section{Non-Uniform Assignments}\label{sec:sa}
\subsection{Weighted LTO placebo test}

Outside of experimental settings, the assumption of uniform treatment is questionable. To handle non-uniform treatment assignments, we first propose a weighted analog of the LTO $p$-value. Suppose that the probability of treating unit $k$ is given by $\pi_k, k =1,\dots,N$. In some applications, these probabilities can be estimated. For example, when assessing the effect of earthquakes \cite{cavallo2013catastrophic, dupont2015happened}, the probability of earthquake in each region can be estimated via statistical seismology \cite{vere2005statistical}. Then the following weighted approximate $p$-value generalizes the LTO $p$-value of Section \ref{sec:ltojk}. 

\begin{definition}
[Weighted LTO $p$-values]
\begin{align*}
    p_{\wLTO}(\pi) & :=  \sum_{\substack{j \neq k\\ j,k \neq I}} \frac{\pi_j \pi_k}{(1-\pi_I)^2 - \sum_{l\neq I} \pi_l^2} \mathbf{I}\lbrace I \not\succ j,k \rbrace 
\end{align*}
\end{definition}

Note that $p_{\wLTO}(\pi) = p_{\nLTO}$ when $\pi_i = 1/N$ for all $i$.
We prove the following upper bound for the Type-I error of $p_{\wLTO}(\pi)$ as the value of an integer programming problem. The proof is deferred to Appendix \ref{sec:sa_proofs}.

\begin{theorem}\label{thm:TypeIerror_weighted_LTO}
For any $\alpha \in (0, 1)$, $\Prob_{H_0}(p_{\wLTO}(\pi) \le \alpha)$ is upper bounded by the value of the following integer programming problem: 
\begin{align*}
\max \,\, & L_1(z)\\
\text{s.t.} \,\, 
& \frac{1}{3}L_1(z)^3 - L_1(z)^2 + L_1(z) \left[ L_2(z) + \alpha \left(1 - \sum_{l=1}^N \pi_l^2\right) \right] \ge  \left(\frac{4}{3}-2\alpha\right)L_3(z) - (1-2\alpha)L_2(z),\\
& L_k(z) = \sum_{i=1}^{N}\pi_i^k z_i, \,\,z_i \in \{0, 1\}, \,\, i = 1,\ldots, N, \,\, k = 1, 2, 3.
\end{align*}
\end{theorem}
When $\pi_i \equiv 1/N$, it can be shown via tedious algebra that Theorem \ref{thm:TypeIerror_weighted_LTO} recovers Theorem \ref{thm:TypeIerror}. In the general non-uniform case, while the bound is given by a complicated nonconvex integer programming problem, it can be solved approximately by off-the-shelf solvers when $N$ is small, e.g., $N < 50$.

\subsection{Sensitivity analysis}
When assignment probabilities are not controlled by the researcher, it is useful to have procedures to analyze sensitivity of our approach to the uniform assignment assumption. The goal of sensitivity analysis is to answer the question: by how much would our statistical conclusions change if the treatment were non-uniform? Phrased another way, by how much would the treatment probabilities have to change to overturn a conclusion of statistical significance? Firpo and Possebom \cite{firpo2018synthetic} give an approach to do sensitivity analysis for the exact placebo test under non-uniform treatments. We can extend their approach to the naive LTO placebo test using the weighted LTO placebo $p$-values as the building block. The procedure is as follows.

\begin{enumerate}
    \item Compute $p_{\nLTO}$. If $p_{\nLTO} > \alpha$, claim insignificance; otherwise, proceed with the following steps.
    \item Fix a $\Gamma \in [1,N].$ Solve the optimization problem
    \begin{gather}
    \label{eq:optimization_sa}
     \text{maximize}_{\pi} \ p_{\wLTO}(\pi) \\
     \text{subject to } \pi_i \in \left[\frac{1}{\Gamma N}, \frac{\Gamma}{N}\right] \,\, \forall i,\,\,\, \text{and }\sum_i \pi_i = 1.\label{eq:optimization_sa_constaint}
    \end{gather}
    \item Find the smallest value of $\Gamma$ by a binary search such that the solution to the above optimization problem is greater than $\alpha.$
\end{enumerate}

\begin{remark}
The choice of the bound $[\frac{1}{\Gamma N}, \frac{\Gamma}{N}]$ is not essential. We can handle any bounds on $(\pi_1, \ldots, \pi_N)$ defined by a single parameter. For example, the uncertainty set considered in the sensitivity analysis of \cite{firpo2018synthetic} can be equivalently formulated as $\max_{i}\pi_i / \min_{i}\pi_i \le \Gamma$ for some $\Gamma \ge 1$.
\end{remark}

By Theorem \ref{thm:TypeIerror_weighted_LTO}, it remains to further add the constraint \eqref{eq:optimization_sa_constaint} into the integer programming problem. By carefully analyzing it, we end up with a much simpler result. The proof is presented in Appendix \ref{sec:sa_proofs}.

\begin{theorem}
\label{thm:weighted_alphaless1/n_typeIerror}
Assume that $\alpha \le 1/3$, $\Gamma \le N/4$, and $\pi_i \in [\frac{1}{\Gamma N}, \frac{\Gamma}{N}]$ for all $i$. Then 
\begin{align*}
\Prob_{H_0}(p_{\wLTO}(\pi)\le \alpha)\leq\frac{3 - \frac{3\Gamma}{N} - \sqrt{9(1-\frac{\Gamma}{N})^2 - 12\left(-\frac{4\Gamma^2}{3N^2} + \frac{\Gamma}{N} + \alpha\left(1- \frac{2\Gamma + 1}{N} + \frac{2\Gamma^2}{N^2}\right)\right) }}{2}.  
\end{align*}
\end{theorem}

When $\Gamma = 1$, Theorem \ref{thm:weighted_alphaless1/n_typeIerror} recovers Theorem \ref{thm:TypeIerror}. When $N$ is large, the bound has the same limit $\frac{3 - \sqrt{9-12\alpha}}{2}$ as obtained for the $\Gamma = 1$ setting. Unlike Theorem \ref{thm:TypeIerror}, we cannot sharpen it by taking advantage of the discreteness of the Type-I error because $\pi_i$ can take a range of values.

\paragraph{Computational Details.} The optimization problem in step 2) may pose challenges. To optimize this problem, we use the following trick. Fixing a $\Gamma$, let $B_\Gamma$ be the set of $\pi$ satisfying the constraints in Equation \eqref{eq:optimization_sa}. We only need to check if $\max_{\pi \in B_\Gamma} p_{\wLTO}(\pi)$ is greater than $\alpha$. This is equivalent to checking that
\[
\max_{\pi \in B_\Gamma} \sum_{\substack{j,k \neq I \\ j \neq k}} \pi_j \pi_k\mathbb{I}\set{I \not\succ j,k} - \alpha\left((1-\pi_I)^2 - \sum_{l\neq I} \pi_l^2\right) > 0
\]
The optimization problem can be written as a simple quadratic program. Namely, we must check whether
\[
\max_{\pi \in B_\Gamma} \pi^\top A\pi + 2\alpha \pi^\top e_I > \alpha
\]
where $A := (G + \alpha I_N -2\alpha e_{I}e_I^\top)$, with $G \in \bR^{N \times N}$ is a zero-one matrix with entries $G_{jk} = \mathbb{I}\set{I \not\succ j,k},$ whenever $j$ and $k$ do not equal $I$ and zero otherwise. Here $e_I$ is the standard basis vector for index $I.$ If the goal is to calculate $\max_{\pi \in B_\Gamma} p_{\wLTO}(\pi)$, we can do a binary search over $c \in [0,1]$ that satisfy $\max_{\pi \in B_\Gamma} p_{\wLTO}(\pi) > c$, and apply the same trick as above. \\

The objective is not necessarily concave as $A$ is not guaranteed to be positive semidefinite. As it turns out, the resulting optimization problem is a fundamental NP-hard problem known as \textit{quadratic nonconvex programming} with box constraints. In practice, non-convex quadratic programming problems are solved using heuristics such as branch-and-bound techniques. As the dimensionality of the problems that we are solving is rather small, we expect off-the-shelf solvers should do well on this problem.  Moreover, whenever the computed LTO $p$-value is small, the matrix $G + I_N - 2e_{I}e_I^\top$ is usually rather sparse, which would help with computation.\\

\section{Simulation Results}
\label{sec:simulations}

Thus far we have explored theoretical properties of the LTO placebo tests. In this section, we compare empirical properties of the LTO placebo tests on semisynthetic examples based on the California Proposition 99 on tobacco consumption \cite{abadie2010synthetic}, the Basque Country terrorism on local economic prosperity \cite{abadie2003economic}, and the German Reunification on West Germany's GDP. \cite{abadie2015comparative}. All simulation results rely on the R package \textsf{Synth} \cite{abadie2011synth}, and may be found at \url{https://github.com/tsudijon/LeaveTwoOutSCI}. 

The structure of the semisynthetic simulations is the same for each dataset. Firstly, we subsample $N$ units (smaller than the size of the original data) from each dataset, excluding the original treated unit. As described briefly in the introduction, on each subsample of the data, a unit is chosen uniformly at random to be treated. For each treated unit $I$, we posit a time-invariant treatment effect $\tau$, which we add on to the observed outcomes $Y_{I,t}$ following $T_0$. A range of values of $\tau$ are used, corresponding approximately to $(0,-1,-2,-3)$ times the standard deviation of the outcome variable measured at $T_0.$  For each subsample, several methods are compared in terms of power: the exact placebo, approximate placebo, randomized placebo, the naive LTO, and powered LTO.

The setting for $\tau=0$ amounts to comparing the Type-I errors for all methods. We compare the methods in two regimes: one in which $\alpha < 1/N$ and one in which $\alpha \geq 1/N$. In the $\alpha < 1/N$ regime case we omit the results for the exact placebo, which will be powerless in this setting. Each method uses the widely-adopted RMSPE statistic defined in \eqref{eq:RMSPE}, and the synthetic controls are constructed in ways which follow the original analyses of each dataset. Finally, we average the powers and Type-I errors for all methods over 20 Monte Carlo subsamples of the original dataset, and report the resulting mean. This is an estimate of the unconditional Type-I error and power of each method.

Broadly, we observe the following empirical results. The LTO procedures demonstrate higher power in large signal to noise ratios, compared against the three placebo $p$-values. In intermediate signal to noise ratios, the naive LTO procedure often has lower or comparable power than the placebo procedures, but the powered LTO fares better. In the $\alpha < 1/N$ setting, we observe that the unconditional Type-I errors of the LTO methods are below $\alpha$ in all cases, whereas that of the approximate placebo $p$-value is always $1/N$. This suggests that the upper bound on the conditional Type-I error established in Theorem \ref{thm:TypeIerror} is not tight, and in practice, the LTO procedures often satisfy Type-I error constraints. The randomized placebo test does control the Type-I error, but its power is quite poor relative to the other procedures.

\subsection{Semisynthetic Power Simulations}
\label{sec:powersimulations}

\paragraph{California Proposition 99.} Before showing semisynthetic results for other datasets, we give further details for the experiments done in Section \ref{sec:intro_example}. In Section \ref{sec:intro}, only the naive LTO procedure was compared with the placebo. A full comparison with both versions of the LTO procedure can be found in Figure \ref{fig:smoking_full_powercomparison}. When constructing the synthetic controls, we matched on several covariates: average retail price of cigarettes, log per capita income, the proportion of the population age 15–24, and per capita beer consumption. We use the mean of these variables over the 1980–1988 period, along with lagged smoking consumption in 1975, 1980, and 1988. This replicates the specification of \cite{abadie2010synthetic}.

\begin{figure}[h]
  \begin{subfigure}{0.5\textwidth}
    \centering
    \includegraphics[width = \linewidth]{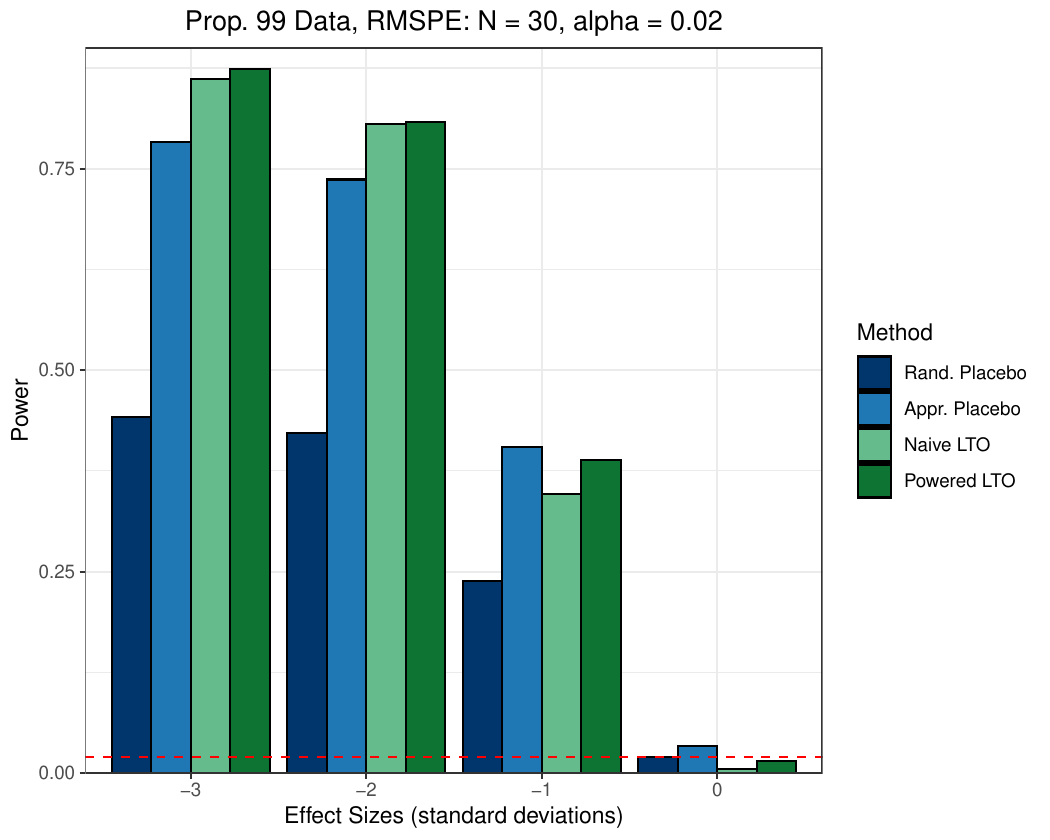}
  \end{subfigure}%
  \begin{subfigure}{0.5\textwidth}
    \centering
    \includegraphics[width=\linewidth]{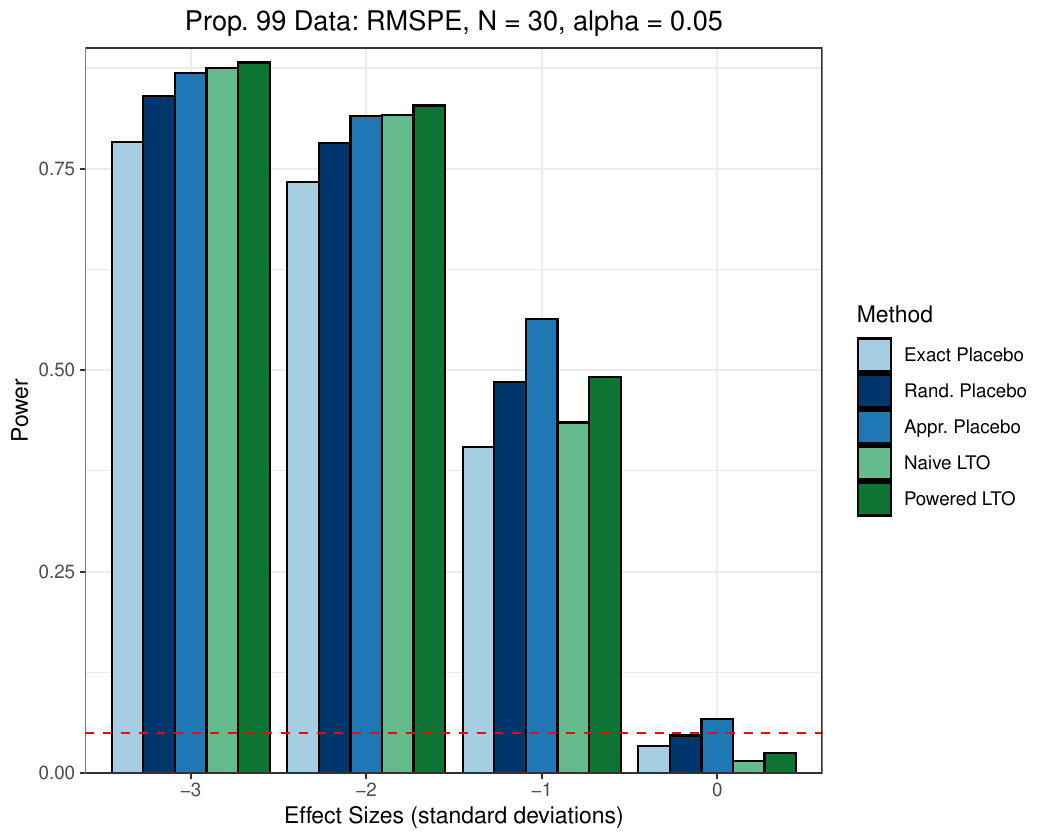}
  \end{subfigure}
  \caption{The same simulation as Figure \ref{fig:smoking_intro_powercomparison} is carried out. Both versions of placebo and LTO procedures are compared in this experiment. Red dashed line indicates the level $\alpha$.\label{fig:smoking_full_powercomparison}}
\end{figure}

 On the left, the powered LTO procedure is unconditionally valid in this experiment, while providing a power boost over the naive procedure. This is enough to make the power comparable to that of the approximate placebo, for intermediate signal to noise ratios. The power gain for larger effect sizes is smaller. The same observations can be seen on the right. While the powered LTO procedure increases Type-I error, its unconditional Type-I error is still less than that of the exact placebo. It is thus unconditionally valid in this experiment. It provides a significant power boost, however. Both LTO methods have larger power than the exact placebo for even modest effect sizes. In this setting, the approximate placebo violates the Type-I error constraint.

\paragraph{Basque Country Terrorism.} Next, we consider a similar setup using a dataset of Spanish regional GDP studied in \cite{abadie2003economic}. In this setting, we take $\alpha = 0.05, N = 15$ as an example of inference in the $\alpha < 1/N$ setting. When constructing all synthetic controls, we imitate the choice of covariates found in the running Basque data example of \cite{abadie2011synth}. The results are found in Figure \ref{fig:basque_rmspe_results}.

\begin{figure}[h]
  \begin{subfigure}{0.5\textwidth}
    \centering
    \includegraphics[width=\linewidth]{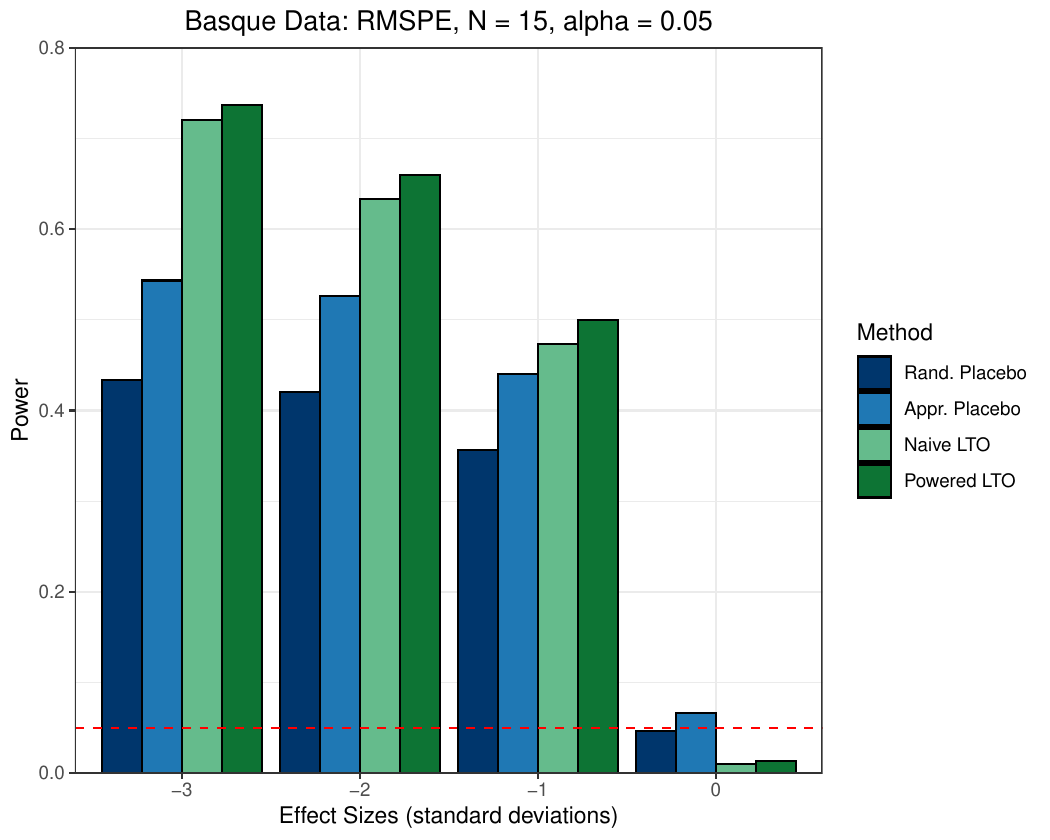}
    \label{fig:basquedata_alpha0.05N15}
  \end{subfigure}%
  \begin{subfigure}{0.5\textwidth}
    \centering
    \includegraphics[width = \linewidth]
    {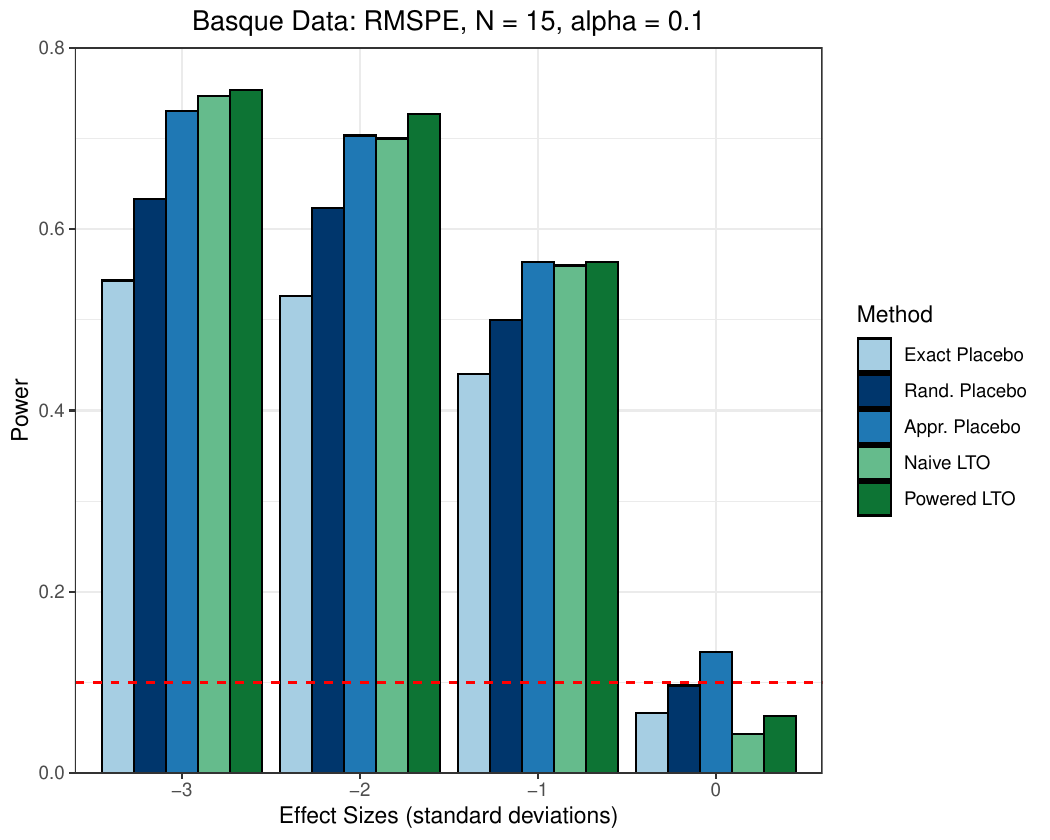}
    \label{fig:basque_valid_rmspe}
  \end{subfigure}
  \caption{Power of several different $p$-values versus effect size, on size 15 subsamples of the Basque country terrorism dataset. The other details are same as Figure \ref{fig:smoking_intro_powercomparison}. (Left) $\alpha = 0.05$, which falls in the $\alpha < 1/N$ regime. The exact placebo is omitted because the power is zero in all cases. (Right) $\alpha = 0.1$, which falls in the $\alpha > 1/N$ regime.}
  \label{fig:basque_rmspe_results}
\end{figure}

Remarkably, the conclusions identified for the Proposition 99 data carry over to this dataset. As before, the LTO procedures demonstrate much lower unconditional Type-I error than the upper bound advertises, while at the same time, they exhibit higher power than the inexact placebo test in higher signal to noise ratios. It is also interesting in this setting that the LTO procedures do well in intermediate signal to noise ratios. 

In the right side of Figure \ref{fig:basque_rmspe_results}, we run the same experiment in when $\alpha = 0.1$; the same observations are broadly true. In this setting, the unconditional Type-I error of the LTO procedures are similar to that of the exact placebo (implying unconditional validity), while the power is much higher in for even moderately large effect sizes. The power of the LTO procedures even matches that of the approximate placebo test, which unconditional Type-I error control here.

\paragraph{German Reunification.} 

In \cite{abadie2015comparative}, the synthetic control method is applied to understand the effects of German reunification on the GDP of West Germany, using 16 other OECD countries as potential controls. The simulation takes multiple size $N = 14$ subsamples of the original data, without West Germany, and compares the power of the placebo and LTO method versus effect size. 

Figure \ref{fig:german_reunification_powercomparison} shows the results for two choices of $\alpha$. Both settings were constructed with the RMSPE statistic. On the left, $\alpha = 0.05$.  In this comparison, the power gain of the LTO procedure is evident at higher signal to noise ratios. In this regime, the unconditional Type-I error of the LTO placebo test is much closer to that of the placebo test, and unfortunately, all procedures with non-zero power must violate conditional validity, per Proposition \ref{prop:no_power}. However, Theorem \ref{thm:TypeIerror} guarantees that the Type-I error will never be worse than that of the placebo test.

On the right, the Type-I errors of the LTO methods are essentially the same as that of the exact placebo, indicating that both methods are unconditionally valid in this setting. It is interesting to notice the power gain of both LTO methods. At intermediate signal to noise ratios, the naive LTO has a similar power to the exact placebo, but for larger signals, the power advantage is considerable. The powered LTO method is considerable more powerful in the intermediate effect size case, and is far more powerful than the exact placebo.

\begin{figure}[h]
  \begin{subfigure}{0.5\textwidth}
    \centering
    \includegraphics[width=\linewidth]{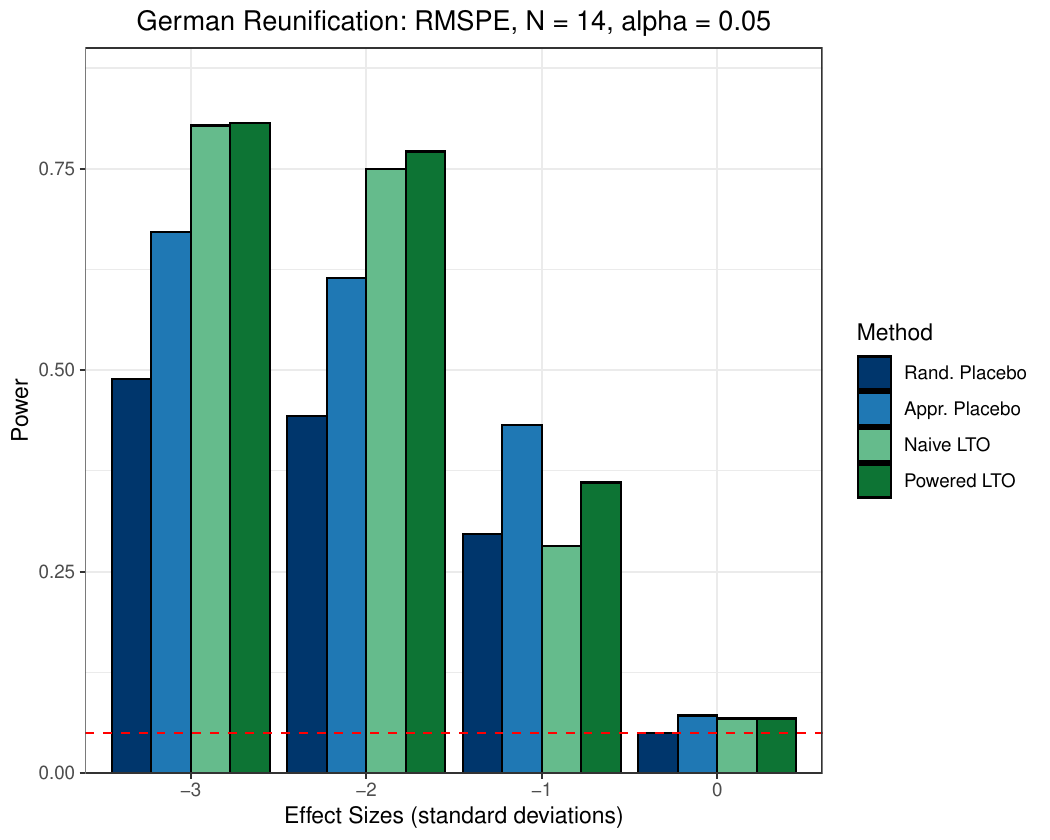}
  \end{subfigure}%
  \begin{subfigure}{0.5\textwidth}
    \centering
    \includegraphics[width = \linewidth]
{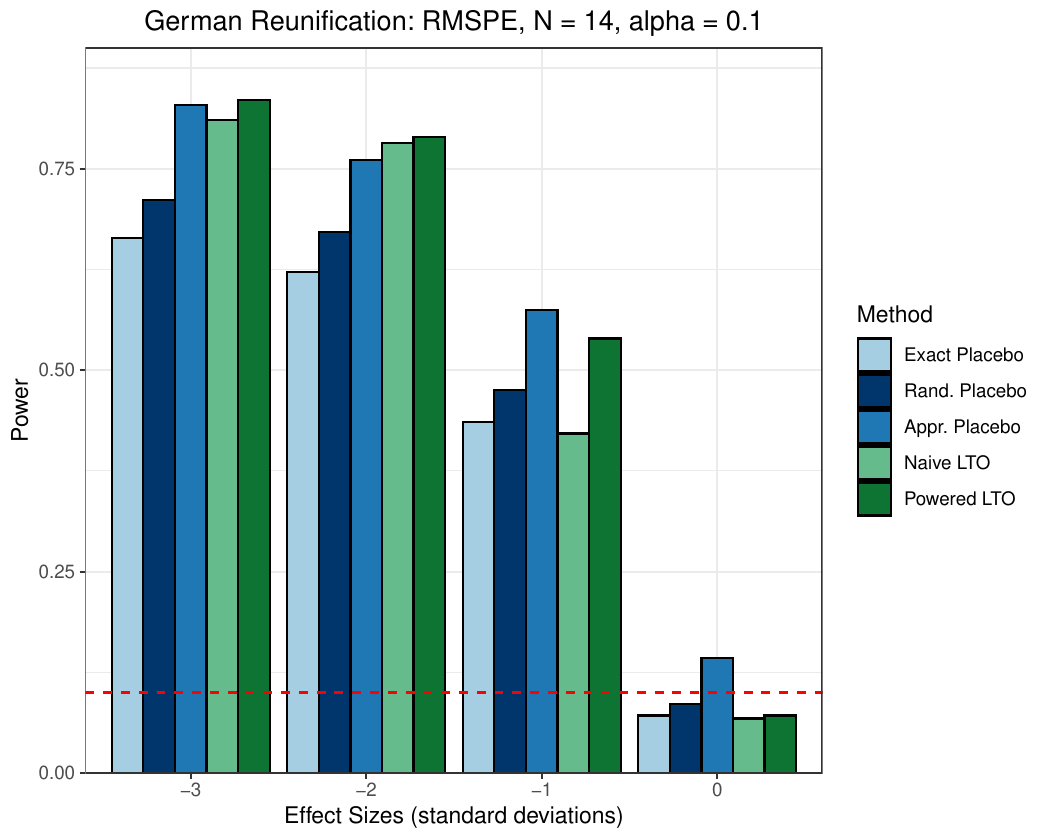}
  \end{subfigure}
  \caption{Power of placebo test and LTO procedure versus effect size, on size 14 subsamples of the German Reunification dataset. The other details are same as Figure \ref{fig:smoking_intro_powercomparison}. (Left) $\alpha = 0.05$, which falls in the $\alpha < 1/N$ regime. The exact placebo is omitted because the power is zero in all cases. (Right) $\alpha = 0.1$, which falls in the $\alpha > 1/N$ regime.}
  \label{fig:german_reunification_powercomparison}
\end{figure}

\subsection{LTO Sensitivity Analyses}

It is worth revisiting the classic case studies of the synthetic control method in \cite{abadie2003economic, abadie2010synthetic, abadie2015comparative} using the LTO and placebo inference procedures. For each use of the LTO $p$-value, we analyze the sensitivity of the conclusion to the equal weights assumption. For all sensitivity analyses, we use Gurobi 10.0 \cite{gurobi} to solve the nonconvex quadratic optimization problem for optimizing $p_{\wLTO}(\pi)$. A similar sensitivity analysis for the exact and approximate placebo test may be done using the results of \cite{firpo2018synthetic}, which we do not replicate here.

\begin{table}[htbp]
  \centering
  \caption{The table shows a summary of LTO and placebo tests and standard placebo tests applied to the case studies in \cite{abadie2003economic, abadie2010synthetic, abadie2015comparative}. The row entitled $\Gamma_{\LTO}$ shows the $\Gamma$ at which the sensitivity analysis for the naive LTO $p$-value overturns the significance. The value for Basque Country is blank, as the $p$-value was not significant at level $\alpha = 0.05$.}
    \begin{tabular}{c|c|c|c|}
    \toprule
    \multicolumn{1}{c}{} & \multicolumn{1}{|c|}{Proposition 99} & \multicolumn{1}{|c|}{Basque Country}
    & \multicolumn{1}{|c|}{German Reunification} \\
    \midrule
    $N$ & $39$ & $17$ & $17$\\
    $p_{\pb}$ & $0$ & $0.35$ & $0$\\
    $p_{\epb}$ & $0.026$ & $0.41$ & $0.059$\\
    $p_{\nLTO}$ & $0.024$ & $0.67$ & $0.042$\\
    $p_{\pLTO}(\alpha)$ & $0.022$ & $0.66$ & $0.03$\\
    $\Gamma_{\LTO}$ & $1.4$ &  $ \text{NA} $ & $1.1$ \\
    \bottomrule
    \end{tabular}%
  \label{tab:sa_results_all}%
\end{table}

In Section \ref{sec:intro}, we already saw the application of the LTO placebo test and standard placebo tests for the Proposition 99 example of \cite{abadie2010synthetic}, along with the sensitivity analysis for the LTO placebo test. In this section, we replicate the original analyses of \cite{abadie2015comparative, abadie2003economic}, using the exact and approximate placebo $p$-value, the naive LTO placebo test, and the powered LTO placebo test. Focusing on the German Reunification dataset \cite{abadie2015comparative}, the exact and approximate placebo $p$-value are both too large to reject the null hypothesis. By contrast, both LTO p-values are below $\alpha$. The correction used by the powered LTO $p$-value is decent. Since $p_{\nLTO}\le \alpha$, we can perform sensitivity analysis to calculate the smallest $\Gamma$ that overturns the conclusion. The output is shown in Figure \ref{fig:sensitivity_analyses}. From the figure, $\Gamma$ is around $1.1$ before the maximum possible $p$-value is greater than $0.05$. Overall, the sensitivity curve does not increase too steeply as a function of $\Gamma$. 

For the Basque country terrorism dataset of \cite{abadie2003economic}, we repeat the exercise to analyze the effect of terrorism on GDP of the Basque country. Here, $N = 17$ with the onset of the treatment being $1970.$ We conduct inference with $\alpha = 0.05,$ which is less than $1/N.$ Surprisingly in this example, we obtained a LTO $p$-value of $0.67$; the placebo $p$-value gave a similar non-significant result at $0.41.$ We expected much smaller $p$-values. One possible explanation is that in prior analyses of the Basque country terrorism dataset \cite{abadie2011synth}, several regions were removed if they had poor fit for the treatment period. We do not incorporate these adjustments in the LTO inference procedure. Because the LTO $p$-value is non-significant, in practice a sensitivity analysis procedure would not be run. It is still interesting to see how the curve of maximum weighted $p$-values grows with $\Gamma$. Figure \ref{fig:sensitivity_analyses} shows the output for this example and the previous one.

\begin{figure}[htbp]
\centering 
  \begin{subfigure}{0.42\textwidth}
    \centering
    \includegraphics[width=\linewidth]{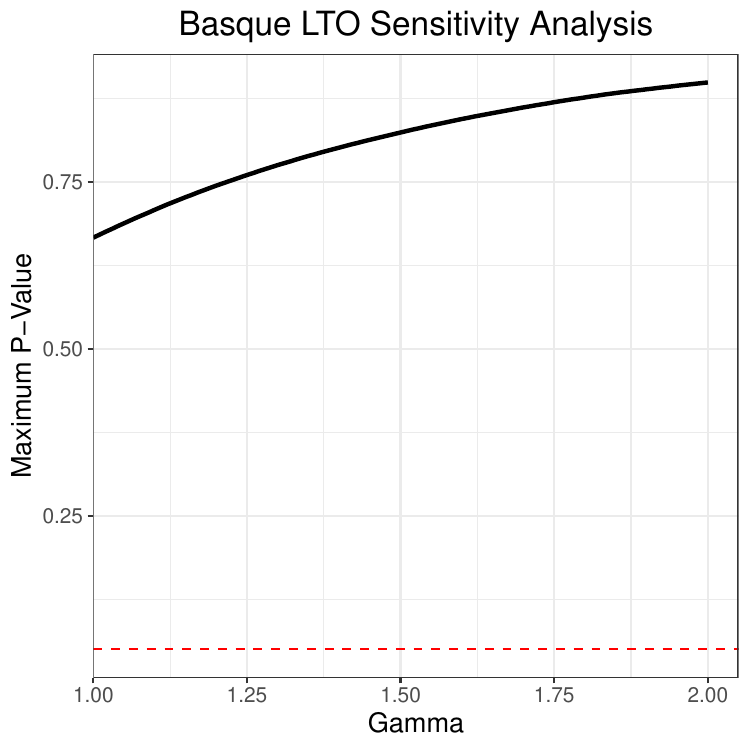}
    \label{fig:basque_sa}
  \end{subfigure}%
  \hspace{0.1\textwidth}
  \begin{subfigure}{0.42\textwidth}
    \centering
    \includegraphics[width = \linewidth]
    {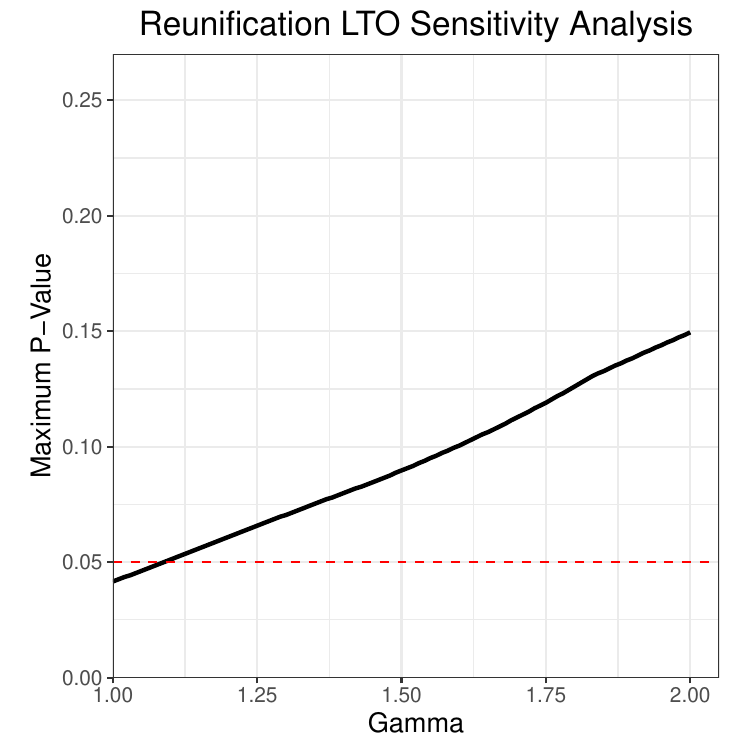}
    \label{fig:german_reunification_sa}
  \end{subfigure}
  \caption{Curve of maximum weighted LTO $p$-value as a function of $\Gamma$, in the range $[1,2]$. The red dashed line signifies the level $\alpha = 0.05$. (Left) The Basque country terrorism dataset. As the unweighted LTO $p$-value is not significant, in practice this curve would not be computed. (Right) The German Reunification dataset of \cite{abadie2015comparative}.}
  \label{fig:sensitivity_analyses}
\end{figure}

\section{Extensions}

We discuss several extensions of the LTO placebo test. The first defines the notion of strange point in a different way, while the second leaves $r$ points out at a time.

\subsection{Rank-sum LTO placebo test}
\label{sec:ranksum}

Instead of labeling a unit as strange if it wins too many matches, we can consider a unit to be strange if its rank across all matches is too large. This leads to a variant of the procedure which has similar coverage guarantees as the original Jackknife+ procedure; see Theorem 1 of \cite{barber2021predictive}. Define $\rank\set{I; I,i,j} :=  \mathbb{I}\set{I \succ j} + \mathbb{I}\set{I \succ i},$ where in this context, $\set{I \succ j} : = \set{R_{I,i,j;I} > R_{I,i,j;j}}$ and similarly for $\set{I \succ i}$. We can define a $p$-value 
\[
p_{\rLTO} := \frac{1}{(N-1)(N-2)}\sum_{\substack{i,j \in [N]\setminus I\\ i \neq j}} (2 - \rank\set{I; I,i,j}),
\]

Using a similar argument as in Theorem \ref{thm:TypeIerror}, we derive an upper bound on the Type-I error of $p_{\rLTO}$. The proof is presented in Appendix \ref{sec:proof_ranksum}. 
\begin{theorem}[Rank-sum LTO placebo test]
\label{thm:ranksum_guarantee}
The Type-I error of $p_{\rLTO}$ has the following bound
\[
\Prob(p_{\rLTO} \leq \alpha) \leq  \frac{\floor{(N- 
1)\alpha} + 1}{N}.
\]

\end{theorem}

The bound is equal to or $1/N$ larger than that of $p_{\nLTO}$ and $p_{\pLTO}(\alpha)$. In particular, when $\alpha < 1/(N-1)$, all three choices have Type-I error controlled at $1/N$, as shown in Theorem \ref{cor:TypeIguarantee_alphaless1/n}. However, it is easy to see that $2 - \rank\set{I; I,i,j}\ge \mathbb{I}\{I\not\succ i,j\}$ and thus $p_{\wLTO}\le p_{\rLTO}$ almost surely. Thus, $p_{\rLTO}$ is less favorable because it always gives more conservative inference.

\subsection{Leave-$r$-out placebo test}
\label{sec:LRO_jk}
We may consider an analogous procedure for the placebo test where $r$-tuples of data points are left out. The analysis of the method in this setting clarifies the essential features of the LTO case. 

Extending the LTO procedure (Remark \ref{rem:tournament}), we consider a tournament between all $N$ units where matches are held for every $(r+1)$ units. For each match among $\{i_0, i_1, \ldots, i_r\}$, we leave out all participants and construct synthetic controls for each of them. The score $R_{\{i_0, i_1, \ldots, i_r\}; i_j}$ of participant $i_j$ is defined via the same form as \eqref{eq:general_score}. For example, we can calculate the RMSPE by using $[N]\setminus\{i_0, i_1, \ldots, i_r\}$ as the donor pool.
We say $i_0$ wins a match among $\{i_0, i_1, \ldots, i_r\}$, denoted by $i_0 \succ i_1, \ldots, i_r$, if $R_{\{i_0, i_1, \ldots, i_r\}; i_0} > \max_{j\in [r]}R_{\{i_0, i_1, \ldots, i_r\}; i_j}$.
Extending the LTO p-value, define
\begin{equation}\label{eq:LRO_pvalue}
p_{\LRO} = \frac{(N-r-1)!}{(N-1)!}\sum_{\substack{i_1, \ldots, i_r\in [N]\setminus I\\ i_1, \ldots, i_r \text{ mutually distinct}}}\mathbb{I}\{I\not\succ i_1, \ldots, i_r\}.
\end{equation}
Anagolously, we can also define the rank-sum LRO p-value:
\begin{equation}\label{eq:LRO_pvalue}
p_{\rLRO} = \frac{(N-r-1)!}{(N-1)!}\sum_{\substack{i_1, \ldots, i_r\in [N]\setminus I\\ i_1, \ldots, i_r \text{ mutually distinct}}}\frac{2}{r}\left(r - \rank\set{I; I, i_1, \ldots, i_r}\right),
\end{equation}
where $\rank\set{I; I, i_1, \ldots, i_r}= \sum_{j=1}^{r}\mathbb{I}\{I\succ i_j\}$ and, in this context, $I\succ i_j$ iff $R_{\{I, i_1, \ldots, i_r\}; I} > \max_{j\in [r]}R_{\{I, i_1, \ldots, i_r\}; i_j}$. Note that they recover $p_{\LTO}$ and $p_{\rLTO}$ when $r = 2$. Above, the normalization factor $2/r$ is chosen to make the Type-I error to converge to $\alpha$ when $N$ tends to infinity, as evidenced by Theorem \ref{thm:ranksum_LRO_guarantee} below. 

The following two results generalize Theorem \ref{thm:TypeIerror} and Theorem \ref{thm:ranksum_guarantee} to $p_{\LRO}$ and $p_{\rLRO}$, respectively. Both proofs are presented in Appendix \ref{sec:LRO_proofs}. Via tedious algebra, we can show that they recover Theorem \ref{thm:TypeIerror} and Theorem \ref{thm:ranksum_guarantee} when $r = 2$.

\begin{theorem}\label{thm:LRO_TypeIerror}
Let 
\[s\mapsto A_{N, r}(s) = \frac{1}{s}\left[\binom{N}{r+1} - \binom{N-s}{r+1}\right].\]
Then $A_{N, r}(\cdot)$ is decreasing on $[N]$ and, for any $\alpha < r/(r+1)$,
\[\Prob_{H_0}(p_{\LRO}\le \alpha)\le \frac{1}{N}\max\left\{k: A_{N, r}(k)\ge (1 - \alpha)\binom{N-1}{r}\right\}.\]
\end{theorem}

\begin{theorem}\label{thm:ranksum_LRO_guarantee}
For any $\alpha \in (0, 1)$,
\[\Prob_{H_0}(p_{\rLRO}\le \alpha)\le \frac{\floor{(N-1)\alpha} + 1}{N}.\]
\end{theorem}

\section{Conclusion and discussion}
\label{sec:discussion}
We propose a novel leave-two-out placebo test for the synthetic control method under uniform assignments. Compared to the standard placebo test based on permutation tests, the LTO placebo test has two advantages. First, it has the same theoretical conditional Type-I error guarantee for practical $(N, \alpha)$ pairs; further, it often achieves a strictly lower unconditional Type-I error in practice, unlike the standard placebo test for which the guarantee is tight.  In particular, this permits valid inference even in the challenging $\alpha < 1/N$ case. Second, we find empirically that it has higher power for moderately large signal-to-noise ratio and comparable power otherwise. In most of our experimental settings, the powered-LTO placebo test almost Pareto-improves the standard placebo test. 

While the design assumption is strong, it does not require large sample size and/or long time series, which are crucial for outcome-based approaches but uncommon in applications of the synthetic control method. Moreover, unlike most outcome-based approaches, the LTO placebo test does not artificially simplify the important steps, like weight matrix estimation and model selection, in the construction of synthetic controls. On the other hand, we propose a weighted variant of the LTO placebo test and a sensitivity analysis to account for non-uniform assignments. Below we discuss a few other related issues.

\subsection{In-time placebo test}
The placebo tests discussed in previous sections all exploit variation in which units are treated. An alternative approach is to exploit the variation of treatment timing. In the literature, in-time placebo tests are often performed by choosing a pre-treatment period, pretending it to be the time when the treatment occurs, and running the same synthetic control procedure to eyeball check if the treated unit differs from its synthetic control between the fake and actual treatment time \cite{fremeth2013making, abadie2015comparative,castillo2017causal, cole2020impact, chattopadhyay2023did}. To formalize it into a rigorous statistical inference procedure, we can follow \cite{firpo2018synthetic} verbatim by assuming the treatment time $T_0$ is uniform in $[t_1, t_2]$ for some $1\le t_0 < t_1 \le T$. Then the standard and LTO placebo tests can both be directly applied in this setting. For in-time placebo tests, the small sample issue is arguably more pronounced because it is typically only reasonable to choose a small time window $[t_1, t_2]$. Based on the promising performance of LTO placebo tests in unit-level placebo tests when the sample size is small, we expect they continue to excel for in-time placebo tests, though we leave the formal investigation for future work. 

\subsection{Testing non-sharp null and constructing confidence regions}
In previous sections we focus on the sharp null \eqref{eq:global_null}. Following \cite{firpo2018synthetic, chernozhukov2021exact}, it is straightforward to test any postulated effect vector $\theta\in \mathbb{R}^{T-T_0}$: 
\[H_0(\theta): \tau_{I, T_0+1} = \theta_1, \tau_{I, T_0+2} = \theta_2, \ldots, \tau_{I, T} = \theta_{T-T_0}.\]
In fact, by subtracting $\theta_j$ from $Y_{I(T_0+j)}(1)$, we are back to the case of sharp null. Phrased differently, we can redefine the RMSPE $R_{i,j,I;I}(\theta)$ for the treated unit within the triple $\{i,j,I\}$ by replacing $Y_{It}$ by $Y_{It} - \theta_{t-T_0}$ for $t=T_0+1, \ldots, T$ and reject $H_0(\theta)$ if $p_{\nLTO}(\theta)\le \alpha$, where
    \begin{equation}
    \label{eq:lto_pvalue_theta}
    p_{\nLTO}(\theta) := \frac{1}{(N-1)(N-2)} \sum_{\substack{i,j \in[N] \setminus I \\ i \neq j}} \mathbb{I}\set{R_{i,j,I;I}(\theta) \le \max(R_{i,j,I;i}, R_{i,j,I;j})},
    \end{equation}
Note that the other RMSPEs remain the same because the synthetic control weights only depend on pre-treatment outcomes which are unchanged.

By inverting the test, we can obtain a confidence region for $(\tau_{I, T_0+1}, \tau_{I, T_0+2}, \ldots, \tau_{I, T})$: 
\[\mathcal{C}_{1-\alpha} = \left\{\theta\in \mathbb{R}^{T-T_0}: p_{\nLTO}(\theta) > \alpha\right\}.\]
When the RMSPE is used, by \eqref{eq:lto_pvalue_theta}, there exists $A_{ij}\in \mathbb{R}^{T-T_0}$ and $b_{ij}\in \mathbb{R}$ for $i,j\in [N]\setminus I$ such that
\[\mathcal{C}_{1-\alpha} = \left\{\theta\in \mathbb{R}^{T-T_0}: \sum_{i,j\in [N]\setminus I}\mathbb{I}\set{\|\theta - A_{ij}\|_2^2 \le b_{ij}} > \alpha (N-1)(N-2)\right\}.\]
Each indicator gives a ball and $\mathcal{C}_{1-\alpha}$ is an union of intersections of these balls that lie in at least $\alpha (N-1)(N-2)$ balls. By Theorem \ref{thm:TypeIerror}, 
\[\Prob\Big((\tau_{I, T_0+1}, \tau_{I, T_0+2}, \ldots, \tau_{I, T})\in \mathcal{C}_{1-\alpha}\Big)\ge 1 - \frac{\floor{N f(N, \alpha)}}{N}.\]

\subsection{Consistency of LTO and inconsistency of approximate placebo test}\label{sec:consistency}
Power analysis is challenging in our setting as it typically relies on large sample approximation that tends to be uninformative when $N$ and $T$ are small, particularly given the complexity of the synthetic control procedure. In this section, we make an attempt to explain the empirically observed power improvement over the approximate placebo test.

In the classical asymptotic regime where $N$ tends to infinity, a test is consistent if its power converges to $1$ for any fixed alternative. Qualitatively, consistency ensures that the null hypothesis will be rejected with high probability as the problem becomes sufficiently easy.
We migrate the notion of consistency to the fixed-$N$ setting by defining a test to be consistent if the power converges to $1$ as the signal strength under the alternative grows to infinity. Consistency serves as a minimal requirement for a desirable test. We show that, if $\alpha < 1/N$, the LTO placebo test is uniformly consistent, whereas the approximate placebo test is not. In fact, the latter may have nearly zero power even with an arbitrarily large signal strength. 

\begin{theorem}\label{thm:consistency}
Let $G$ denote the distribution of $(Y_{i,t}(0))_{i\le N, t\le T}$. Consider an alternative hypothesis parametrized by a single parameter $\eta > 0$:
\[H_{1,\eta}: Y_{i,t}(1) = Y_{i, t}(0) + \eta \tau_{i,t}\]
for some deterministic matrix $\tau = (\tau_{i, t})_{i\le N, t\le T} \in \Gamma$ where $\Gamma$ includes all matrices $\tau$ with $\min_{i}\sum_{t}\tau_{it}^2 \ge 1$. Assume $N, T$ are fixed and $\alpha < 1/N$. Consider the LTO and approximate placebo tests based on the post-RMSPE statistic defined in \eqref{eq:RMSPE}. Then, for any tight collection of distributions $\mathcal{G}$ on $\mathbb{R}^{N\times T}$, 
\[\lim_{\eta \rightarrow \infty}\inf_{G\in \mathcal{G}, \tau \in \Gamma}\Prob_{G, \tau,  H_{1,\eta}}(p_{\nLTO} \le \alpha) = 1,\]
Furthermore, there exists infinitely many pairs $(G, \tau) \in \cl{G} \times \Gamma$ that do not vary with $\eta$ such that 
\[
\lim_{\eta \rightarrow \infty}\Prob_{G, \tau, H_{1,\eta}}(p_{\pb} \le \alpha) \left\{\begin{array}{cc}
= 0 & \text{ if $N$ is even}\\
\le 1/N & \text{if $N$ is odd.}
\end{array}\right.
\]
\end{theorem}

A proof may be found in Section \ref{sec:proofs_discussion}.

\subsection{Leave-two-out techniques for other problems}\label{sec:few_clusters}

The theory of LTO placebo tests only relies on the uniform assignment assumption but not the choice of summary statistics. Thus, it can be applied beyond synthetic control applications to any problem with the same structure. Here, we discuss one instance in the context of inference with few clusters \cite{bertrand2004much,
donald2007inference, ibragimov2010t,
conley2011inference, bester2011inference, imbens2016robust,
ibragimov2016inference,canay2017randomization, 
mackinnon2018wild, hagemann2019placebo,
canay2020use, hagemann2020inference,
canay2021wild, hagemann2023permutation, webb2023reworking}; see also the surveys by \cite{cameron2015practitioner} and \cite{mackinnon2023cluster}. With a few exceptions that assumes normality of errors (e.g., \cite{donald2007inference}), existing approaches with theoretical guarantees rely on large clusters in order for the within-cluster test statistics to be close to normally distributed (e.g., \cite{ibragimov2010t, ibragimov2016inference, hagemann2020inference}) or at least close to be symmetrically distributed (e.g., \cite{canay2017randomization, canay2021wild, hagemann2023permutation}) -- these distributional assumptions are hard to justify in most applications of synthetic controls.  Other procedures work well empirically when the cluster sizes are moderate, e.g., $\ge 30$ \cite{cameron2008bootstrap}, though the theory still requires the sample size tends to infinity in each cluster. Moreover, to do inference on treatment effects, many methods that can accommodate small clusters require at least two \cite{conley2011inference, hagemann2019placebo} or $\log (1/\alpha)$ \cite{canay2017randomization} treated clusters if the total number of clusters is below $1/\alpha$. 

In sum, when there is a single treated cluster, which is common in applications, existing methods rely on either large-sample approximation, or that the number of control clusters is at least $1/\alpha - 1$. In this case, if the data comes from a cluster randomized experiment \cite{su2021model, lu2023design} or if the researcher is willing to assume uniform assignment, the LTO procedure can be applied to test Fisher's sharp null. We leave the details for future research.

\subsection{Leave-two-out test as a new type of randomization inference}\label{sec:beyond}
Classical randomization inference \cite{fisher1937design, romano1989bootstrap, romano1990behavior, ding2017paradox} exploits exchangeability, or group invariance more generally, by comparing the test statistic for the observed data with those for permuted or group-transformed data. Despite a century of development, nearly all randomization inference procedures are variants of this generic form with different choices of test statistics or group of transformations. The LTO placebo tests, or the more general leave-$r$-out tests, are new types of randomization inference that move beyond classical permutation or rank-based inference. 

Our tests generalize the Jackknife+ \cite{barber2021predictive}, which also differs from classical randomization inference. It is an extension of conformal prediction \cite{vovk2005algorithmic} and the Jackknife procedure \cite{wu1986jackknife} to construct prediction intervals with guaranteed coverage under exchangeability. 
Let $Z_1 = (X_1, Y_1), \ldots, Z_N = (X_N, Y_N)$ be exchangeable observations where $X_i$ denotes a set of covariates and $Y_i$ denotes the outcome. The goal is to construct a prediction interval $\hat{C}(X_{N+1})$ such that $\Prob(Y_{N+1} \in \hat{C}(X_{N+1}))\ge 1-\alpha$ assuming $Z_{N+1} = (X_{N+1}, Y_{N+1})$ is exchangeable with $\{Z_1, \ldots, Z_N\}$. Jackknife+ essentially inverts a leave-one-out test:
\[p(Y_{N+1}) = \frac{1}{N}\sum_{i=1}^{N}\mathbb{I}\set{R_{i,N+1;N+1}\le R_{i, N+1; i}},\]
where $R_{i,j;j} = |Y_j - \hat{\mu}_{-i,-j}(X_j)|$ where $\hat{\mu}_{-i,-j}(\cdot)$ can be any model fit on $(Z_k)_{k\in [N]\setminus \{i, j\}}$. However, Jackknife+ suffers from the same granularity issue as the standard placebo tests because it can only take multiples of $1/N$.

Another exception is the prediction inference of model transfer errors by \cite{andrews2022transfer}. They consider a setting with data from $N$ exchangeable training domains (e.g., countries), denoted as $Z_1, \ldots, Z_N$. For a given economic or machine learning model and a pair of domains $Z_i, Z_j$, let $e(Z_i, Z_j)$ denote the accuracy (e.g., the mean square error) of the model fit on the $i$-th domain and evaluated on the $j$-th domain. Let $Z_{N+1}$ denote an unobserved target domain that is exchangeable with $\{Z_1, \ldots, Z_N\}$. The goal is to construct a prediction interval for the random transfer error $e(Z_I, Z_{n+1})$ where $I$ is chosen from $[N]$ uniformly at random. Their interval is constructed based on quantiles of $\{e(Z_i, Z_j): i,j\in [N], i\neq j\}$ even if they are no longer exchangeable as in classical randomization inference due to the bivariate structure.

\bibliographystyle{alpha}
\bibliography{main}

\newpage
\appendix

\section{A Survey of Applications of Synthetic Controls}
\label{sec:smalldata}

In this section, we conduct a literature survey on empirical applications of synthetic controls. The results are summarized in the Tables \ref{tab:lit_review_numbers} and \ref{tab:lit_review_design}. We focused on the size of the datasets in the studies, along with the use of placebo inference. Several findings are of interest:
\begin{itemize}
    \item \textbf{Small Datasets}. For the overwhelming majority of studies surveyed, the synthetic control was used to analyze one treated unit, with a donor pool on the order of a few dozen. In table \ref{tab:lit_review_numbers}, several studies use a staggered adoption framework, although the synthetic control method is still applied to analyze one treated unit at a time. In only a few studies \cite{liu2015spillovers, ando2015dreams, acemoglu2016value} was the donor pool near $100$ or more, which enables inference at traditional levels like $\alpha = 0.01$. Furthermore, the number of pre- and post-treatment time periods is typically also on the order of a few dozen. In only a handful of studies \cite{fremeth2013making, chattopadhyay2023did, liu2015spillovers, castillo2017causal, acemoglu2016value, bell2023texas} were the number of time periods substantial. On the other hand,  several studies used the placebo test when $N_0,T_0$ or $T_1$ were less than $5$ \cite{coffman2012hurricane, bohn2014did}. Asymptotic inferential procedures seem difficult to justify in these settings. 
    \item \textbf{Resulting Issues with Placebo Inference}. The predominant robustness check /  inference methodology used among papers in the survey was the unit-level placebo test. Only a few papers did not use the placebo procedure. Moreover, some authors note that the placebo-based $p$-value cannot be created when using small levels $\alpha$, and more broadly, convey the difficulties in doing inference when the donor pool is small. For example, \cite{chattopadhyay2023did} notes ``The two-sided RI $p$-value cannot get any smaller than $0.04$ [\dots]. Thus, the coefficients of interest in this study cannot be significant at a 1\% level of significance." Newiak and Williams \cite{newiak2017evaluating} remark similarly: ``Typically, however, there are not many countries in the control group for which we can construct counterfactuals with a fit that is at least as good as for the treated country, which limits the informativeness of this exercise (and prevents us from calculating meaningful p-values)."
    \item \textbf{Alternative Inference Procedures}. Alternative inferential procedures (methods to calculate $p$-values) were few. Several papers \cite{adhikari2016evaluating, almer2012effect} in the staggered adoption setting use a placebo method for average treatment effects for grouped data outlined in \cite{cavallo2013catastrophic}. Otherwise, in-time placebo robustness checks were occasionally used \cite{fremeth2013making, abadie2015comparative,castillo2017causal, cole2020impact, chattopadhyay2023did, donohue2019right, ando2015dreams, ben2022synthetic}. A few papers attempted a bootstrap-based approach \cite{sills2015estimating, ben2022synthetic}, justified using asymptotic theory with assumptions on the data generating process.
\end{itemize}

\begin{table}
\centering

\begin{tabular}{|c||c|c|c|c|c|c|c|}
\hline
Paper & Units & $N_1$ & $N_0$ & Time Unit & $T_0$ & $T_1$ & Used placebo \\ \hline \hline
\cite{bell2023texas} & State & 1 & 50 & Monthly & 75 & 9 & Yes \\ \hline
\cite{cunningham2018decriminalizing} & State & 1 & 50 & Year & 43 & 8 & Yes \\ \hline
\cite{dustmann2017labor} & District & 1 & 85 & Year & 4 & 6 & Yes \\ \hline
\cite{kleven2013taxation} & Country & 1 & 13 & Year & 11 & 13 & Yes \\ \hline
\cite{pinotti2015economic} & Region & 5 & 15 & Year & 25 & 27 & Yes \\ \hline
\cite{trejo2024effects} & District & 1 & 54 & Year & 8 & 4 & No \\ \hline
\cite{hankins2020finally} & State & 1 & 46 & Year & 22 & 10 & Yes \\ \hline
\cite{coffman2012hurricane} & County & 1 & 3 & Year & 17 & 16 & Yes \\ \hline
\cite{fremeth2013making} & Car Company & 1 & 16 & Month & 95 & 35 & Yes \\ \hline
\cite{castillo2017causal} & Province & 1 & 19 & Month & 90 & 125 & Yes \\ \hline
\cite{peri2019labor} & City & 1 & $\leq 43$ & Year & 6 & 11 & Yes \\ \hline
\cite{dupont2015happened} & Prefecture & 1 & 41 & Years & 19 & 14 & Yes \\ \hline
\cite{boes2012effect} & City & 1 & $35$ & Year & 8 & 18 & Yes \\ \hline
\cite{grier2016economic} & Country & 1 & $\leq 20$ & Year & 28 & $\leq 10$ & Yes \\ \hline
\cite{chamon2017fx} & Country & 1 & 16 & Week & 12 & 12 & Yes \\ \hline
\cite{cole2020impact} & City & 1 & 29 & Day & 30 & 12 & Yes \\ \hline
\cite{roesel2017mergers} & German State & 1 & 10 & Year & 9 & 5 & No \\ \hline
\cite{sills2015estimating} & Municipality & 1 & 35 & Year & 5 & 5 & No \\ \hline
\cite{bohn2014did} & State & 1 & 46 & Year & 8 & 2 & Yes \\ \hline
\cite{chattopadhyay2023did} & State & 1 & 49 & Month & 136 & 53 & Yes \\ \hline
\cite{ben2021augmented} & State & 1 & 49 & Quarter & 89 & 16 & Yes \\ \hline
\cite{donohue2019right}* & State & 33 & $\leq 33$ & Years & 10 & $\leq 40$ & No \\ \hline
\cite{billmeier2013assessing}* & Country & 30 & $\leq 50$ & Year & 10 & $\leq 40$ & Yes \\ \hline
\cite{adhikari2016evaluating}* & Country & 8 & $\leq 30$ & Year & 10 & $\leq 15$ & Yes \\ \hline
\cite{ando2015dreams}* & Japanese Muni. & 8 & $\leq 90$ & Year & $\leq 30$ & $\leq 15$ & Yes \\ \hline
\cite{munasib2015regional}* & State & 3 & 29 & Year & $\leq 6$ & $\leq 6$ & Yes \\ \hline
\cite{newiak2017evaluating}* & Countries & 7 & 39 & Year & $\leq 20$ & $\leq 10$ & Yes \\ \hline
\cite{liu2015spillovers}* & US Counties & 57 & $\leq 1180$ & Year & $\leq 30$ & $\leq 80$ & Yes \\ \hline
\cite{almer2012effect}* & Country & 7 & 20 & Year & 17 & 6 & Yes \\ \hline
\cite{ben2022synthetic}* & State & 32 & $\leq 49$ & Year & $\leq 28$ & $\leq 25$ & No \\ \hline
\cite{acemoglu2016value}* & Stocks & 22 & 561 & Day & 250 & 30 & Yes \\ \hline
\end{tabular}
\caption{Summary of a broad literature review of papers using the SCM, found through Google Scholar. $N_1,N_0$ denote the number of treated units and donor units, while $T_0,T_1$ denote the number of pre- and post-treatment periods in the study. These numbers reflect our best estimate of the study design from reading the corresponding papers. Some papers conduct multiple analyses using different outcome variables, for which the size of the donor pool or time periods change. In these cases, we provide a rough upper bound on the figures. * denotes studies with a staggered adoption framework: the studies comprise multiple applications of the synthetic control method with a single treated unit. As a result, the number of units in the donor pool and the number of pre- and post-treatment periods change per analysis. For example, \cite{billmeier2013assessing} study the impact of economic liberalization episodes on GDP per capita across a dataset of all countries. The synthetic control analysis done for Singapore used 42 placebos, while Indonesia had $8$ placebos. }
\label{tab:lit_review_numbers} 
\end{table}

\begin{table}[H]
\centering
\begin{tabular}{|p{1.75cm}||p{3cm}|p{6cm}|p{5cm}|}
\hline
Paper & Units & Treatment & Outcome  \\ \hline \hline
\cite{bell2023texas} & US States & TX 2021 Abortion Ban & Number of Births \\ \hline
\cite{cunningham2018decriminalizing} & US States + D.C. & RI decriminalizing prostitution & Sexual crimes \& health outcomes \\ \hline
\cite{dustmann2017labor} & German Districts & German-Czech Labor Market Shock & Employment measures \\ \hline
\cite{kleven2013taxation} & European Country & Bosman Ruling (Tax Policy for Athletes) & Athlete emigration\\ \hline
\cite{pinotti2015economic} & Italian Region & Organized Crime Rise & Various economic indicators \\ \hline
\cite{trejo2024effects} & MI School District & Flint, Michigan Water Crisis & Educational achievement \\ \hline
\cite{hankins2020finally} & US State & Nebraska's 1937 Legislative Structure Change & Government expenditure  \\ \hline
\cite{coffman2012hurricane} & County & Hawaii's 1992 Hurricane Iniki & Economic development \\ \hline
\cite{fremeth2013making} & Car companies & Government support for Chrysler & Car sales  \\ \hline
\cite{castillo2017causal} & Province & Tourism policy in Salta, Argentina & Employment  \\ \hline
\cite{peri2019labor} & City & Miami 1980 Mariel Boatlift & Labor market  \\ \hline
\cite{dupont2015happened} & Japanese Prefecture & 1995 Kobe Earthquake & Economic conditions  \\ \hline
\cite{boes2012effect} & City & Economic liberalization in Mumbai & Skilled wages  \\ \hline
\cite{grier2016economic} & Country & Chavez Regime in Venezuela & Economic conditions \\ \hline
\cite{chamon2017fx} & Country & Brazil FX intervention & Exchange rates \\ \hline
\cite{cole2020impact} & City & Wuhan COVID 19 lockdown & Air pollutant concentration  \\ \hline
\cite{roesel2017mergers} & German State & Govt. merger in Saxony & Govt. spending  \\ \hline
\cite{sills2015estimating} & Municipality & Local enforcement initiatives in Paragominas, Brazil & Deforestation  \\ \hline
\cite{bohn2014did} & US State & 2007 Legal Arizona Workers Act & Unauthorized immigrant pop.  \\ \hline
\cite{chattopadhyay2023did} & US State & MA Health Reform Program & Self-employment  \\ \hline
\cite{ben2021augmented} & US State &  Personal income tax cut & GDP per capita \\ \hline
\cite{donohue2019right}* & US State & Right-to-carry laws & Violent crime  \\ \hline
\cite{billmeier2013assessing}* & Country & Economic liberalization & GDP per capita \\ \hline
\cite{adhikari2016evaluating}* & Country & Flat tax system & GDP per capita  \\ \hline
\cite{ando2015dreams}* & Japanese Muni. & Establishment of nuclear power facilities & Local per capita income  \\ \hline
\cite{munasib2015regional}* & US State & Oil and gas development & Economic development  \\ \hline
\cite{newiak2017evaluating}* & Countries & IMF advising-monitoring program & Economic development  \\ \hline
\cite{liu2015spillovers}* & US County & 1862 Morrill Act & Economic development \\ \hline
\cite{almer2012effect}* & Country & Binding emission targets under the Kyoto Protocol & Domestic CO2 emissions \\ \hline
\cite{ben2022synthetic}* & US State & Teacher collective bargaining laws & Student expenditures \& teacher salary \\ \hline
\cite{acemoglu2016value}* & Financial Stocks & Geithner Confirmation Leak & Cumulative Stock Returns \\ \hline
\end{tabular}
\caption{Additional detail on design studies for papers in the literature review.}
\label{tab:lit_review_design} 
\end{table}

\newpage

\section{Proofs}
\label{sec:proofs}

\subsection{An impossibility result on conditionally valid p-values}\label{subapp:sec1}

Proposition \ref{prop:no_power} gives the formal statement of the impossibility result mentioned in Section \ref{sec:intro}.

\begin{proposition}\label{prop:no_power}
Suppose $I$ is uniform on $[N]$ and let $\alpha < 1/N$. Then any conditionally finite-sample-valid, non-randomized $p$-value $p(I)$ has zero power: $\Prob_{H_1}(p(I) \leq \alpha | \cl{Y}) = 0,$ for  any alternative $H_1: \tau_{i,T_0 + 1} = \theta_1, \ldots,\tau_{i,T} = \theta_{T-T_0}$, for all $i =1,\dots,N$. 
\end{proposition}

Here we write $\tau_{i,t} = Y_{i,t}(1) - Y_{i,t}(0)$ and recall that $\cl{Y}$ denotes the set of control potential outcomes.

\begin{proof}[Proof of Proposition \ref{prop:no_power}]
Let $p(I)$ be such a $p$-value, written as a function of the index of treated unit $I$. We may also write $p(I,\cl{D})$ to emphasize the panel dataset $\cl{D}$ that is used to create the $p$-value. Let $\cl{D}^{(i)}$ be equal to the panel dataset of control potential outcomes, with the vector $(\theta_1, \ldots, \theta_{T-T_0})$ added on the unit $i$ in the post-treatment periods. Under $H_1,$ the treated unit $I$ has observed data $Y_{I,t} = Y_{I,t}(0) + \theta_{t}$, while for all other units $i$, we observe $Y_{i,t}(0).$ Conditioning on the value of $I,$
\begin{equation}
\Prob_{H_1}(p(I) \leq \alpha |\cl{Y}) = \frac{1}{N}\sum_{i \in [N]}\mathbf{1}\set{p(i, \cl{D}^{(i)}) \leq \alpha}.
\end{equation}

Now, consider the conditional rejection probability under the sharp null $\Prob_{H_0}(p(I) \leq \alpha | \cl{Y}).$ This latter probability can be written as $\Prob_{H_0}(I\in \mathcal{A})$ for some event $\mathcal{A}\subset [N]$. When $I$ is uniform on $[N] = \{1, \ldots, N\}$, $\Prob_{H_0}(p(I) \leq \alpha |\cl{Y})$ belongs to the set $\set{0/N,1/N,\dots,N/N}$. If $\alpha < 1/N, \Prob_{H_0}(p(I) \leq \alpha | \cl{Y}) = 0$ because $p$ is conditionally valid. Thus
\begin{equation}
0 = \Prob_{H_0}(p(I) \leq \alpha | \cl{Y}) = \frac{1}{N} \sum_{i = 1}^N \mathbf{1}\set{p(i,\cl{Y}) \leq \alpha}
\end{equation}
for every choice of $\cl{Y}$. Summing the above over $\cl{D}^{(i)}$, 
\[
\frac{1}{N} \sum_{j = 1}^N \sum_{i = 1}^N \mathbf{1}\set{p(i,\cl{D}^{(j)}) \leq \alpha} = 0,
\]
which implies 
\begin{equation}
    \Prob_{H_1}(p(I) \leq \alpha |\cl{Y}) = \frac{1}{N}\sum_{i \in [N]}\mathbf{1}\set{p(i, \cl{D}^{(i)}) \leq \alpha}  \leq \frac{1}{N} \sum_{j = 1}^N \sum_{i = 1}^N \mathbf{1}\set{p(i,\cl{D}^{(j)}) \leq \alpha}  = 0,
\end{equation}
as claimed, since all summands are non-negative.
\end{proof}

\subsection{Proofs for Section \ref{sec:ltojk}}\label{subapp:sec2}

\begin{proof}[Full Proof of Thm. \ref{thm:TypeIerror}] 

Let us start with the decomposition of the sum outlined in Section \ref{sec:ltojk}:
\begin{align*}
     & \overbrace{\sum_{i,j,k \in \cl{S}}  \bb{I}(k \succ i,j)}^{\text{(I)}}
    +  \overbrace{2\sum_{k \in \cl{S}}\sum_{i \in \cl{S},j \in \cl{S}^c} \bb{I}(k \succ i,j)}^{\text{(II)}}
    +  \overbrace{\sum_{k \in \cl{S}}\sum_{i,j \in \cl{S}^c}  \bb{I}(k \succ i,j)}^{\text{(III)}}. 
\end{align*}

For the first sum, note by renaming the labels that 
$\sum_{i,j,k \in \cl{S}} \bb{I}(k \succ i,j) = \sum_{i,j,k \in \cl{S}}  \bb{I}(i \succ j,k) = \sum_{i,j,k \in \cl{S}}  \bb{I}(j \succ k,i).$ Thus 
\[
\text{(I)} \leq \sum_{i,j,k \in \cl{S}} \frac{1}{3}\left(\bb{I}(k \succ i,j) + \bb{I}(i \succ j,k) + \bb{I}(j \succ k,i) \right) \leq s(s-1)(s-2)/3.
\]
For the second sum, swap the naming of the $i,k$ labels and use the bound $\bb{I}(k \succ i,j) + \bb{I}(i \succ k,j) \leq 1$ to see that 
\[
\text{(II)} \leq 2\sum_{k \in \cl{S}}\sum_{i \in \cl{S},j \in \cl{S}^c} \frac{1}{2}\left(\bb{I}(k \succ i,j) + \bb{I}(i \succ k,j) \right) \leq s(s-1)(N - s).
\]
Finally, in the last sum, we use the naive bound $\text{(III)} \leq s (N-s)(N-1-s).$  \\

Combining these bounds, canceling a factor of $s,$ we find that 
\begin{align}
(1-\alpha) (N-1)(N-2) &\leq (s-1)(s-2)/3 + (s-1)(N - s) +  (N-s)(N-1-s)\nonumber\\
& = (s-1)(s-2)/3 + (N-s)(N-2)\label{eq:LTO_key_inequality}
\end{align}

The right hand side reduces to
\begin{align*}
\frac{1}{3}s^2 + \frac{2}{3} - s(N-1) - 2N + N^2.
\end{align*}
Thus
\[
\frac{1}{3}s^2 + \frac{2}{3} - s(N-1) - 2N + N^2 - (N^2 - 3N + 2) + \alpha (N-1)(N-2) \geq 0,
\]
which reduces to

\[
\frac{s^2}{3} - \frac{4}{3} - s(N-1) + N + \alpha (N-1)(N-2) \geq 0
\]
Introducing $\beta = s/N$, we may reduce to the quadratic inequality 
\[
\frac{\beta^2}{3} - \beta(1 - \frac{1}{N}) - \frac{4}{3N^2} + \frac{1}{N} + \alpha(1 - \frac{1}{N})(1-  \frac{2}{N}) \geq 0.
\]
This is a parabola that opens upwards. When $\beta = 1$, since $\alpha < 2/3$, the value is upper bounded by 
\[\frac{1}{3} - 1 + \frac{2}{N} - \frac{4}{3N^2} + \frac{2}{3}(1 - \frac{1}{N})(1-  \frac{2}{N}) = 0.\]
Thus, $\beta$ must be upper bounded by the smaller root of the quadratic function. On the other hand, when $\beta = 1/N$, the value is $\alpha(1 - 1/N)(1-2/N) > 0$. Thus, the smaller root is positive. We can solve this root by quadratic formula:

\begin{equation}
\beta \leq \frac{3 - \frac{3}{N} - \sqrt{9(1-\frac{1}{N})^2 - 12\left(-\frac{4}{3N^2} + \frac{1}{N} + \alpha(1 - \frac{1}{N})(1 - \frac{2}{N})\right) }}{2}. 
\end{equation}
To conclude, we note that $\Prob(p_{\nLTO} \leq \alpha) = \Prob(I \in \cl{S}) = \beta.$ The bound in equation \eqref{eq:quadraticbound_jkproof} completes the proof. \\
\end{proof}

\begin{proof}[Proof of Theorem \ref{cor:TypeIguarantee_alphaless1/n}]

Since $N \ge 2$, $1/N < 2/3$. To prove the first claim, it is left to prove that \eqref{eq:LTO_key_inequality} does not hold for $s = 2$ when $\alpha \le \frac{1}{N-1}$. Plugging $s = 2$ into \eqref{eq:LTO_key_inequality}, we have 
\[(1 - \alpha)(N- 1)(N-2)\le (N-2)^2\Longrightarrow \alpha \ge \frac{1}{N - 1}.\]
To prove the second claim, we first note that $4/N < 2/3$ when $N > 6$. Thus, it is left to show that \eqref{eq:LTO_key_inequality} does not hold with $s = \floor{N\alpha}+2$. We consider three cases separately. 

\begin{itemize}
\item When $\alpha \in [0, 1/N)$, the result has been proved in the first claim.
\item When $\alpha \in [1/N, 2/N)$, \eqref{eq:LTO_key_inequality} with $s = 3$ implies 
\[(1 - \alpha)(N - 1)(N - 2)\le \frac{2}{3} + (N-2)(N-3)\Longrightarrow \alpha\ge \frac{2N - 14/3}{(N - 1)(N - 2)} > \frac{2}{N},\]
where the last inequality uses that $N > 6 > 3$.
\item When $\alpha \in [2/N, 3/N)$, \eqref{eq:LTO_key_inequality} with $s = 4$ implies 
\[(1 - \alpha)(N - 1)(N - 2)\le 2 + (N-2)(N-4)\Longrightarrow \alpha\ge  \frac{3N - 8}{(N-1)(N-2)} > \frac{3}{N},\]
where the last inequality uses that $N > 6$.

\item When $\alpha \in [3/N, 4/N\cdot (1 - 2/(N-1)(N-2)))$, \eqref{eq:LTO_key_inequality} with $s = 5$ implies 
\[(1 - \alpha)(N - 1)(N - 2)\le 4 + (N-2)(N-5)\Longrightarrow \alpha\ge  \frac{4N - 12}{(N-1)(N-2)} = \frac{4}{N}\left(1 - \frac{2}{(N-1)(N-2)}\right)  .\]

\end{itemize}

Finally, we establish the third claim. For $k \ge 4$, let $h(k) = \frac{k}{N-1}\left(\frac{N}{N-2}\frac{k-1}{3} - 1\right)$. Then by assumption $h(k) < 1$. Write $N\alpha$ as $k - \Delta$. It is left to prove that \eqref{eq:LTO_key_inequality} does not hold with $s = k+1$ when $\Delta \in (h(k), 1)$. In fact, the LHS minus the RHS is equal to
\begin{align*}
&(1 - (k - \Delta)/N)(N-1)(N-2) - \frac{k(k-1)}{3} - (N-k-1)(N-2)\\
& = \left(\Delta + \frac{k - \Delta}{N}\right)(N - 2) - \frac{k(k-1)}{3}\\
& = \frac{(N-1)(N-2)}{N}(\Delta - h(k)) > 0.
\end{align*}

\end{proof}

\subsection{Proofs for Section \ref{sec:sa}}\label{sec:sa_proofs}

\begin{proof}[Proof of Theorem \ref{thm:TypeIerror_weighted_LTO}]
We first analyze the refined weighted $p$-value. Because $\sum_{\substack{i \neq j\\ i,j \neq I}} \pi_i \pi_j = (1-\pi_I)^2 - \sum_{l\neq I} \pi_l^2$, the event that the $p$-value is less than or equal to $\alpha$ is the same as the event that 
\[
\sum_{\substack{i,j \neq I \\ j \neq i}} \frac{\pi_i \pi_j}{(1-\pi_I)^2 - \sum_{l\neq I} \pi_l^2 }\mathbb{I}\set{I \succ i,j} \geq 1-\alpha
\]
Define a unit $k$ to be strange if it satisfies the above bound in place of $I,$ i.e., 
\[
\sum_{\substack{i,j \neq k \\ j \neq i}} \frac{\pi_i \pi_j}{(1-\pi_k)^2 - \sum_{l\neq k} \pi_l^2 }\mathbb{I}\set{k \succ i,j} \geq 1-\alpha
\]
Call $\cl{S}$ the set of strange points and $s = |\cl{S}|$. Let $\beta = \sum_{i\in \cl{S}}\pi_i = \Prob(p_{\wLTO}(\pi)\le \alpha)$. Then multiply on both sides by $\pi_k((1-\pi_k)^2 - \sum_{l\neq k} \pi_l^2)$ and sum over $k \in \cl{S}$ to obtain the inequality
\begin{equation}
    \sum_{k \in \cl{S}}\sum_{\substack{i,j \neq k \\ j \neq i}} \pi_i\pi_j\pi_k\mathbb{I}\set{I \succ j,k} \geq (1-\alpha) \left( \sum_{k \in \cl{S}} \pi_k(1-\pi_k)^2 - \sum_{k \in \cl{S}, l\neq k} \pi_k\pi_l^2 \right).
\end{equation}
Throughout the rest of the proof, we omit the terms $i\neq j,$ and $ i,j\neq k$ below the summations for notational convenience. That is, all summations over two or more indices requires all indices to be mutually distinct. The left hand side of the above may be split into three terms. The sums are over distinct indices.
\begin{align*}
     & \overbrace{\sum_{i,j,k \in \cl{S}} \pi_k \pi_i\pi_j \bb{I}(k \succ i,j)}^{\text{(I)}}
    +  \overbrace{2\sum_{k \in \cl{S}}\sum_{i \in \cl{S},j \in \cl{S}^c} \pi_k \pi_i\pi_j \bb{I}(k \succ i,j)}^{\text{(II)}}
    +  \overbrace{\sum_{k \in \cl{S}}\sum_{i,j \in \cl{S}^c} \pi_k \pi_i\pi_j \bb{I}(k \succ i,j)}^{\text{(III)}}. 
\end{align*}
For sum (I), split the sum into three copies of itself, and permute the labels cyclically. That is, by switching the indexing labels, we note $\sum_{i,j,k \in \cl{S}} \pi_k \pi_i\pi_j \bb{I}(k \succ i,j) = \sum_{i,j,k \in \cl{S}} \pi_k \pi_i\pi_j \bb{I}(i \succ j,k) = \sum_{i,j,k \in \cl{S}} \pi_k \pi_i\pi_j \bb{I}(j \succ k,i)$. Next, because at most one of the residuals $R_{i,j,k;\cdot}$ can be the largest, $\bb{I}(i \succ j,k) + \bb{I}(j \succ k,i) + \bb{I}(k \succ i,j)  \leq 1.$ This leads to the bound
\begin{align*}
\text{(I)} \leq \frac{1}{3}\sum_{i,j,k \in \cl{S}} \pi_i\pi_j\pi_k & = \frac{1}{3}\left(\beta^3 - 3\sum_{i,j\in \cl{S}} \pi_i^2\pi_j - \sum_{i\in \cl{S}} \pi_i^3\right) = \frac{1}{3}\left(\beta^3 - 3\beta\sum_{i\in \cl{S}} \pi_i^2  + 2\sum_{i\in \cl{S}} \pi_i^3\right)
\end{align*}

For the second sum (II), switching the labels of $i,k$ gives the identity $$\sum_{k \in \cl{S}}\sum_{i \in \cl{S},j \in \cl{S}^c} \pi_k \pi_i\pi_j \bb{I}(k \succ i,j) = \sum_{i \in \cl{S}}\sum_{k \in \cl{S},j \in \cl{S}^c} \pi_k \pi_i\pi_j \bb{I}(i \succ k,j).$$ Bounding the indicator sum $\bb{I}(k \succ i,j) + \bb{I}(i \succ k,j)$ by one, we have the bound
\[
\text{(II)} \leq  \sum_{i \in \cl{S}}\sum_{k \in \cl{S},j \in \cl{S}^c} \pi_k \pi_i\pi_j   = \left(\beta^2 - \sum_{i \in \cl{S}} \pi_i^2\right)(1 - \beta).
\]
Lastly, we bound (III) by the naive bound; replacing the indicator by the constant $1$, the sum can be bounded by $\beta[(1 - \beta)^2 - \sum_{i \in \cl{S}^c} \pi_i^2].$\\

Summing these bounds, we find
\begin{align*}
\frac{1}{3}\beta^3  - \beta^2 + \beta\left(1 - \sum_{i \in \cl{S}^c}\pi_i^2\right) + \frac{2}{3} \sum_{i \in \cl{S}} \pi_i^3 - \sum_{i\in \cl{S}}\pi_i^2 & \geq (1-\alpha) \left( \sum_{k \in \cl{S}} \pi_k(1-\pi_k)^2 - \sum_{k \in \cl{S}, l\neq k} \pi_k\pi_l^2 \right) \\
& = (1-\alpha)\left( \beta - 2\sum_{k \in \cl{S}} \pi_k^2 + 2\sum_{k \in \cl{S}} \pi_k^3 - \beta \sum_{l=1}^N \pi_l^2 \right).
\end{align*}

After some simplification, the resulting bound gives
\begin{equation}\label{eq:weighted_TypeIerror_key}
\frac{1}{3}\beta^3 - \beta^2 + \beta \left[ \sum_{i \in \cl{S}} \pi_i^2 + \alpha \left(1 - \sum_{l=1}^N \pi_l^2\right) \right] \geq \left(\frac{4}{3}-2\alpha\right)\sum_{i\in \cl{S}}\pi_i^3 - (1-2\alpha)\sum_{i \in \cl{S}} \pi_i^2.
\end{equation}
The proof is completed by replacing $\beta, \sum_{i\in\mathcal{S}}\pi_i^2, \sum_{i\in \mathcal{S}}\pi_i^3$ with $L_1(z), L_2(z), L_3(z)$, respectively.

\end{proof}

\begin{proof}[Proof of Theorem \ref{thm:weighted_alphaless1/n_typeIerror}]
By the Cauchy-Schwarz inequality, $\sum_{\ell=1}^{N}\pi_\ell^2\ge 1/N$. Thus, \eqref{eq:weighted_TypeIerror_key} implies that
\[\frac{1}{3}\beta^3 - \beta^2 + \beta \alpha \frac{N-1}{N}\geq \sum_{i \in \cl{S}}\pi_i\left[\left(\frac{4}{3}-2\alpha\right)\pi_i^2 - (1 + \beta -2\alpha)\pi_i\right].\]
Let $g(x) = \left(4/3-2\alpha\right)x^2 - (1 + \beta -2\alpha)x$. Then $g'(x) = 2\left(4/3-2\alpha\right)x - (1 + \beta -2\alpha)$. Since $\beta\ge 0$ and $\alpha \le 1/3$, for any $x\le 1/4$, $g'(x) \le 2/3 - \alpha - (1 - 2\alpha) = \alpha - 1/3 \le 0$. Thus, $g(x)$ is decreasing on $[0, 1/4]$. Since $\Gamma\le N/4$, $\pi_i\le 1/4$. Using $\pi_i \leq \frac{\Gamma}{N}$, we find that 
\[\frac{1}{3}\beta^3 - \beta^2 + \beta \alpha \frac{N-1}{N}\geq \sum_{i \in \cl{S}}\pi_i g\left(\frac{\Gamma}{N}\right) = g\left(\frac{\Gamma}{N}\right)\beta.\]
Dividing both sides by $\beta$, we have that
\begin{equation}\label{eq:sensitivity_quadratic}
\frac{1}{3}\beta^2 - \beta + \alpha \frac{N-1}{N}\geq \left(\frac{4}{3}-2\alpha\right)\frac{\Gamma^2}{N^2} - (1 + \beta -2\alpha)\frac{\Gamma}{N}
\end{equation}
\[\Longrightarrow \frac{1}{3}\beta^2 - \left(1 - \frac{\Gamma}{N}\right)\beta + \alpha \frac{N-1}{N} - \left(\frac{4}{3}-2\alpha\right)\frac{\Gamma^2}{N^2} + (1 -2\alpha)\frac{\Gamma}{N}\ge 0.\]
Again, this is a parabola that opens upwards. When $\beta = 1$, the LHS is equal to 
\[-\frac{2}{3} + \alpha \frac{N-1}{N} + \tilde{g}\left(\frac{\Gamma}{N}\right), \,\, \text{where }\tilde{g}(x) = (2-2\alpha)x - \left(\frac{4}{3} - 2\alpha\right)x^2.\]
For any $x\le 1/4$, $\tilde{g}'(x) = (2 - 2\alpha) - 2(4/3 - 2\alpha)x > (2 - 2\alpha)- (4/3-2\alpha) > 0$. Thus, $\tilde{g}(x)$ is increasing on $[0, 1/4]$ and 
\[-\frac{2}{3} + \alpha \frac{N-1}{N} + \tilde{g}\left(\frac{1}{4}\right) = -\frac{1}{4} + \left(\frac{N-1}{N} - \frac{3}{8}\right) \alpha \le -\frac{1}{4} + \frac{5\alpha}{8} < 0,\]
where the last inequality uses the condition $\alpha \le 1/3$. As a result, \eqref{eq:sensitivity_quadratic} implies that $\beta$ is upper bounded by the smaller root, i.e.,
\[\beta\leq\frac{3 - \frac{3\Gamma}{N} - \sqrt{9(1-\frac{\Gamma}{N})^2 - 12\left(-\frac{4\Gamma^2}{3N^2} + \frac{\Gamma}{N} + \alpha\left(1- \frac{2\Gamma + 1}{N} + \frac{2\Gamma^2}{N^2}\right)\right) }}{2}.\]

\end{proof}

\subsection{Proofs for Section \ref{sec:ranksum}}\label{sec:proof_ranksum}

\begin{proof}[Proof of Theorem \ref{thm:ranksum_guarantee}]
In this proof, call a unit $i$ \textit{strange} if 
\[
\sum_{\substack{j,k \in [N]\setminus i\\ j \neq k}} \frac{1}{2}(\mathbb{I}\set{i \succ j} + \mathbb{I}\set{i \succ k}) \geq (1-\alpha/2)(N-1)(N-2).
\]
Let $\cl{S}$ be the set of strange points. Summing this inequality over all strange points, we find
\[
2(1-\alpha/2)(N-1)(N-2)s \leq \sum_{i \in \cl{S}} \sum_{\substack{j,k \in [N]\setminus i \\ j\neq k}} (\mathbb{I}\set{i \succ j} + \mathbb{I}\set{i \succ k}).
\]
Again, we will assume all indices in the summations are distinct, and omit this from notation. Let us split this sum into three cases. The first case sums over $i, j, k \in \cl{S}$; the second is twice the summation over $i,j \in \cl{S},k \in \cl{S}^\compl$, and the third case sums over $i \in \cl{S}, j,k \in \cl{S}^\compl$. Throughout the rest of the proof, we omit the terms $i\neq j, i,j\neq k$ below the summations for notational convenience. That is, all summations over two or more indices requires all indices to be mutually distinct.\\

The first sum is the quantity $\sum_{i,j,k \in \cl{S}} (\mathbb{I}\set{i \succ j} + \mathbb{I}\set{i \succ k})$. By swapping the naming of the indices $i,j$ and $i,k$, observe that the first sum is equal to
\begin{align*}
    \frac{1}{3}\sum_{i,j,k \in \cl{S}} (\mathbb{I}\set{i \succ j} + \mathbb{I}\set{i \succ k} 
    +  \mathbb{I}\set{j \succ i} + \mathbb{I}\set{j \succ k} 
    + \mathbb{I}\set{k \succ j} + \mathbb{I}\set{k \succ i}).
\end{align*}
The summand is bounded above by $3$. We may thus bound the first sum above by $s(s-1)(s-2)$. The second sum is treated similarly. Swapping the naming of the indices $i,j$, the second sum can be rewritten 
\[
\sum_{i,j \in \cl{S}}\sum_{k \in \cl{S^\compl}}(\mathbb{I}\set{i \succ j} + \mathbb{I}\set{i \succ k} 
    +  \mathbb{I}\set{j \succ i} + \mathbb{I}\set{j \succ k}).
\]
The sum of the indicators on the inside is bounded above by $3$, so the second sum can be bounded by $3s(s-1)(N-s)$. Finally, for the third sum, use the naive bound $2s(N-s)(N-s-1)$.\\

Combining these bounds yields the following inequality on $s$:
\[
2(1-\alpha/2)(N-1)(N-2)s \leq s(s-1)(s-2) + 3s(s-1)(N-s) + 2s(N-s)(N -s-1),
\]
which reduces to $2(1-\alpha/2)(N-1)(N-2) \leq  2N^2 - 5N - (N-2)s + 2$. Thus
\begin{align*}
    (N-2)s & \leq \alpha(N-1)(N-2) - 2(N^2 - 3N + 2) + 2N^2 - 5N + 2\\
    (N-2)s & \leq \alpha(N-1)(N-2) + N - 2.
\end{align*}
Thus
\[
s \leq \alpha(N-1) + 1,
\]
and so $s/N \leq \alpha + \frac{1 - \alpha}{N}.$ Finally, since $s$ is an integer, we have $s\le \floor{N\alpha + 1 - \alpha} = \floor{(N-1)\alpha} + 1$.
\end{proof}

\subsection{Proofs for Section \ref{sec:LRO_jk}}\label{sec:LRO_proofs}

\begin{proof}[Proof of Theorem \ref{thm:LRO_TypeIerror}]
We defer the proof of the monotonicity of $A_{N,r}(s)$ to Lemma \ref{lemma:binomial_coef_identity}. Throughout the proof, for any $k \le r$ and subset $\mathcal{A}\subset [N]$, we denote by $\sum_{i_1, \ldots, i_k\in \mathcal{A}}$ (or $\sum_{i_0, i_1, \ldots, i_k\in \mathcal{A}}$) the sum over all $k$-tuples with mutually distinct elements in $\mathcal{A}$ and by $\sum^{*}_{i_1, \ldots, i_k\in \mathcal{A}}$ (or $\sum^{*}_{i_0, i_1, \ldots, i_k\in \mathcal{A}}$) the sum over all $k$-tuples with strictly increasing elements in $\mathcal{A}$. Define a unit $i_0$ to be strange if 
\begin{equation}
\label{eq:LRO_strange}
    (1-\alpha)\frac{(N-1)!}{(N-r-1)!} \leq \sum_{i_1,\dots,i_r\in [N]\setminus i_0} \mathbb{I}\{i_0 \succ i_1,\dots,i_r\}.
\end{equation}
Further, let $\cl{S}$ be the set of strange points and $s := |\cl{S}|$. Clearly, $p_{\LRO}\le \alpha$ iff $I\in \mathcal{S}$. Then $\Prob_{H_0}(p_{\LRO}\le \alpha) = s/N$. Summing both sides of the equation \eqref{eq:LRO_strange} over strange points, we obtain the bound 
\begin{equation}\label{eq:LRO_strange_sum}
(1-\alpha)s\frac{(N-1)!}{(N-r-1)!} \leq \sum_{i_0\in \mathcal{S}}\sum_{i_1,\dots,i_r\in [N]\setminus i_0} \mathbb{I}\{i_0 \succ i_1,\dots,i_r\}.
\end{equation}
We split the RHS of \eqref{eq:LRO_strange_sum} by $\mathcal{S}\cap \{i_1, \ldots, i_r\}$:
\[\sum_{i_0\in \mathcal{S}}\sum_{i_1,\dots,i_r\in [N]\setminus i_0} \mathbb{I}\{i_0 \succ i_1,\dots,i_r\} = \sum_{J\subset [r]}\sum_{i_j\in \mathcal{S}^c, j\in J^c}\sum_{i_j\in \mathcal{S}, j\in J\cup \{0\}}\mathbb{I}\{i_0 \succ i_1,\dots,i_r\}.\]
Note that the indicator $\mathbb{I}\{i_0\succ i_1, \ldots, i_r\}$ is invariant to reshuffling of $(i_1, \ldots, i_r)$, the summands can be grouped by the size of $J$: 
\begin{equation}\label{eq:LRO_symmetry}
\sum_{J\subset [r]}\sum_{i_j\in \mathcal{S}, j\in J\cup \{0\}}\sum_{i_j\in \mathcal{S}^c, j\in J^c}\mathbb{I}\{i_0 \succ i_1,\dots,i_r\} = \sum_{t=0}^{r} \binom{r}{t}\sum_{i_{t+1}, \ldots, i_r\in \mathcal{S}^c}\sum_{i_0, i_1, \ldots, i_t\in \mathcal{S}}\mathbb{I}\{i_0 \succ i_1,\dots,i_r\}.
\end{equation}
Swapping the naming of $i_0$ with $i_1,i_2,\dots,i_t$ we have the bound
\[
\sum_{i_0, i_1,\dots,i_t \in \cl{S}} \mathbb{I}\set{i_0 \succ i_1,\dots,i_r} =\sum_{i_0, i_1,\dots,i_t \in \cl{S}}  \frac{1}{t+1} \sum_{j=0}^t \mathbb{I}\set{i_j \succ i_1,\dots,i_r} \leq \sum_{i_0, i_1,\dots,i_t \in \cl{S}} \frac{1}{t+1}.
\]
Thus, \eqref{eq:LRO_strange_sum} implies 
\begin{align*}
    & (1-\alpha)s\frac{(N-1)!}{(N-r-1)!} \\
    & \leq r!\sum_{t=0}^r \frac{1}{t!}\frac{1}{(r-t)!}\left(\frac{s!}{(s-t-1)!} \frac{(N - s)!}{(N-s-(r-t))!} \frac{1}{t+1} \right) \\
    & = r!\sum_{t= 0 }^r \binom{s}{t+1}\binom{N-s}{r-t}\\
    & = r!\left[\binom{N}{r+1} - \binom{N-s}{r+1}\right],
\end{align*}
where the last lines applies the Vandermonde's identity. Thus, 
\begin{equation}\label{eq:LRO_strange_final}
A_{N, r}(s) = \frac{1}{s}\left[\binom{N}{r+1} - \binom{N-s}{r+1}\right]\ge (1 - \alpha)\binom{N-1}{r}.
\end{equation}
By definition, 
\[A_{N, r}(1) = \binom{N}{r+1} - \binom{N-1}{r+1} = \binom{N-1}{r} > (1 - \alpha)\binom{N-1}{r},\]
and, for any $\alpha < r/(1+r)$,
\[A_{N,r}(N) = \frac{1}{N}\binom{N}{r+1} = \frac{1}{r+1}\binom{N-1}{r} < (1 - \alpha)\binom{N-1}{r}.\]
By Lemma \ref{lemma:binomial_coef_identity}, $A_{N,r}(s)$ is decreasing. Thus, 
\[s\le \max\left\{k: A_{N, r}(k)\ge (1 - \alpha)\binom{N-1}{r}\right\}.\]
The proof is then completed by noticing that the Type-I error is $s/N$.
\end{proof}

\begin{proof}[Proof of Theorem \ref{thm:ranksum_LRO_guarantee}]
Throughout the proof we use the same notation $\sum$ and $\sum^{*}$ as in the proof of Theorem \ref{thm:LRO_TypeIerror}. Define a unit $i_0$ to be strange if 
\begin{equation}\label{eq:ranksum_LRO_strange}
  r\left(1-\frac{\alpha}{2}\right)\frac{(N-1)!}{(N-r-1)!} \leq \sum_{i_1,\dots,i_r\in [N]\setminus i_0} \rank\{i_0; i_0, i_1, \ldots, i_r\} = \sum_{i_1,\dots,i_r\in [N]\setminus i_0}\sum_{j=1}^{r}\mathbb{I}\{i_0\succ i_j\}.
\end{equation}
Further, let $\cl{S}$ be the set of strange points and $s := |\cl{S}|$. Clearly, $p_{\rLRO}\le \alpha$ iff $I\in \mathcal{S}$. Then $\Prob_{H_0}(p_{\rLRO}\le \alpha) = s/N$. Summing both sides of the equation \eqref{eq:ranksum_LRO_strange} over strange points, we obtain the bound 
\begin{equation}\label{eq:ranksum_LRO_strange_sum}
 r\left(1-\frac{\alpha}{2}\right)s\frac{(N-1)!}{(N-r-1)!} \leq \sum_{i_0\in \mathcal{S}}\sum_{i_1,\dots,i_r\in [N]\setminus i_0}\sum_{j=1}^{r}\mathbb{I}\{i_0\succ i_j\}.
\end{equation}
Following the same argument as \eqref{eq:LRO_symmetry}, we can rewrite the RHS as 
\begin{align*}
&\sum_{t=0}^{r} \binom{r}{t}\sum_{i_{t+1}, \ldots, i_r\in \mathcal{S}^c}\sum_{i_0, i_1, \ldots, i_t\in \mathcal{S}}\sum_{j=1}^{r}\mathbb{I}\{i_0\succ i_j\}\\
& = \sum_{t=0}^{r}\binom{r}{t} \sum_{i_{t+1}, \ldots, i_r\in \mathcal{S}^c}\Bigg( \overbrace{\sum_{i_0, i_1, \ldots, i_t\in \mathcal{S}}\sum_{j=1}^t \mathbb{I}\set{i_0 \succ i_j}}^{\text{(I)}} + \overbrace{\sum_{i_0, i_1, \ldots, i_t\in \mathcal{S}}\sum_{j'=t+1}^r \mathbb{I}\set{i_0 \succ i_{j'}}}^{\text{(II)}} \Bigg).
\end{align*}
Fixing $t$ and $i_{t+1}, \ldots, i_r\in \mathcal{S}^c$, the quantity (I) can be expressed as 
\[\sum_{i_0, i_1, \ldots, i_t\in \mathcal{S}}\rank(i_0; i_0, i_1, \ldots, i_t).\]
Now, we can cyclically relabel $(i_0, i_1, \ldots, i_t)$ to $(i_{k \mod t}, i_{k+1 \mod t}, \ldots, i_{k+t \mod t})$, yielding
\[\sum_{i_0, i_1, \ldots, i_t\in \mathcal{S}}\rank(i_0; i_0, i_1, \ldots, i_t) = \sum_{i_0, i_1, \ldots, i_t\in \mathcal{S}}\frac{1}{t+1}\sum_{j=0}^{t}\rank(i_j; i_0, i_1, \ldots, i_t)\le \sum_{i_0, i_1, \ldots, i_t\in \mathcal{S}}\frac{t}{2},\]
where the last line uses the fact that the sum of ranks is at most $1 + 2 + \ldots + t = (t+1)t/2$. Thus, 
\[\text{(I)}\le \frac{t}{2}\cdot \frac{s!}{(s-t-1)!}.\]
For the quantity (II), we simply bound the indicator by $1$ which yields
\[\text{(I)}\le (r-t)\cdot \frac{s!}{(s-t-1)!}.\]
Putting two terms together, we can upper bound the RHS of \eqref{eq:ranksum_LRO_strange_sum} by
\begin{align*}
    &\sum_{t = 0}^r \binom{r}{t} \frac{s!}{(s - t - 1)!} \frac{(N -s)!}{((N - s) - (r - t))!} \left(r - \frac{t}{2}\right) \\
    & = r!\sum_{t = 0}^r s \binom{s - 1}{t} \binom{N - s}{r - t}\left(r - \frac{t}{2}\right) \\
    & =  sr\cdot r!\sum_{t = 0}^r \binom{s-1}{t} \binom{(N-1) - (s - 1)}{r-t}  - \frac{s(s-1)}{2}r!\sum_{t = 1}^r \binom{s-2}{t-1} \binom{N-1 - (s-1)}{(r-1) - (t-1)}\\
    & =  sr\cdot r!\sum_{t = 0}^r \binom{s-1}{t} \binom{(N-1) - (s - 1)}{r-t}  - \frac{s(s-1)}{2}r!\sum_{t' = 0}^{r-1} \binom{s-2}{t'} \binom{N-1 - (s-1)}{(r-1) - t'}\\
    & = s r!\left(r\binom{N-1}{r} - \frac{s-1}{2}\binom{N-2}{r-1}\right),
\end{align*}
where the last line applies the Vandermonde identity. By \eqref{eq:ranksum_LRO_strange_sum}, 
\[r\left(1 - \frac{\alpha}{2}\right)\binom{N-1}{r}\le r\binom{N-1}{r} - \frac{s-1}{2}\binom{N-2}{r-1}.\]
This implies
\[
s\le \alpha r \binom{N-1}{r}\Big/\binom{N-2}{r-1} + 1 = (N-1)\alpha + 1.
\]
Since $s$ is an integer, $s\le \floor{(N-1)\alpha} + 1$. The proof is then completed.
\end{proof}

\subsection{Proofs for Section \ref{sec:discussion}}
\label{sec:proofs_discussion}

\begin{proof}[Proof of Theorem \ref{thm:consistency}]
We first prove uniform consistency of the LTO placebo test. Since $\mathcal{G}$ is tight, for any $\epsilon > 0$, there exists a constant $M_\epsilon$ such that, 
\[\inf_{G\in \mathcal{G}}P_{G}\left(\max_{i,t}|Y_{it}(0)| \le M_\epsilon\right) \ge 1-\epsilon.\]
On the event $\max_{i,t}|Y_{it}(0)| \le M_\epsilon$, for any $i, j\neq I$ and $k\in \{i, j, I
\}$,
\[\|\hat{\mathbf{Y}}^{-(i,j,I)}_{k}\|_{\infty}\le M_{\epsilon},\]
since it is a convex average of post-treatment outcome vectors. Then 
\[R_{i,j,I;i}\le 4M_{\epsilon}^2, \quad R_{i,j,I;j}\le 4M_{\epsilon}^2.\]
Using the fact that $(a + b)^2 \ge b^2/2 - a^2$, 
\begin{align*}
R_{i,j,I;I} &= \frac{1}{T - T_0}\sum_{t=T_0+1}^T(Y_{I, t}(0) - \hat{Y}_{I, t}^{-(i,j,I)}(0) + \eta \tau_{I, t})^2\\
& \ge \frac{1}{2(T-T_0)}\sum_{t=T_0+1}^T\tau_{I, t}^2 \eta^2 - \frac{1}{T - T_0}\sum_{t=T_0+1}^T(Y_{I, t}(0) - \hat{Y}_{I, t}^{-(i,j,I)}(0))^2\\
& \ge \frac{1}{2(T-T_0)}\eta^2 - 4M_{\epsilon}^2,
\end{align*}
where the last line applies the condition $\tau \in \Gamma$. If $\eta > 4M_{\epsilon}\sqrt{T-T_0}$, 
\[R_{i,j,I;I} > \max\{R_{i,j,I;i}, R_{i,j,I;j}\}.\]
As this holds for all $i, j \neq I$ on the event $\max_{i,t}|Y_{it}(0)| \le M_\epsilon$, 
\[\Prob_{G, \tau, H_{1,\eta}}(p_{\nLTO} = 0) \ge 1 - \epsilon,\]
as long as $\eta > 4M_{\epsilon}\sqrt{T-T_0}$, letting $\epsilon \rightarrow 0$ yields that
\[\lim_{\eta \rightarrow \infty}\inf_{G\in \mathcal{G}, \tau\in \Gamma}\Prob_{G,\tau, H_{1,\eta}}(p_{\nLTO} \le \alpha) = 1.\]

Next, we turn to inconsistency of the approximate placebo test. We start with an even $N$. Let $G$ be any distribution of $(Y_{i, t}(0))$ such that 
\[Y_{2j-1, t}(0) = Y_{2j, t}(0) \triangleq \tilde{Y}_{j, t}, \quad j=1,\ldots, N/2,\]
for some mutually distinct $\tilde{Y}_{j, t}$, and no $\tilde{Y}_{j, t}$ lies in the convex hull of $\{\tilde{Y}_{k, t}: k\neq j\}$ (e.g., $\|\tilde{Y}_{j, t}\| = 1$ for all $j = 1, \ldots, N/2$). Further, let $\tau$ be any matrix in $\Gamma$ such that 
\[\tau_{2j, t} = \tau_{2j+1, t} \triangleq \tilde{\tau}_{j, t}.\]

We prove that, for any such $(G, \tau)$ and any $\eta > 0$, 
\begin{equation}\label{eq:pnLTO=1/N}
\Prob_{G, H_{1,\eta}}(p_{\nLTO} \ge 1/N) = 1.
\end{equation}
This implies $\lim_{\eta\rightarrow \infty}\Prob_{G, \tau, H_{1,\eta}}(p_{\nLTO} \le \alpha) = 0$ when $\alpha < 1/N$. 

To prove \eqref{eq:pnLTO=1/N}, we note that the synthetic control for unit $2j$ must be $2j+1$ and vice versa, as no $\tilde{Y}_{j, t}$ can be expressed as the convex average of others. Thus, 
\[R_{2j} = R_{2j+1} = \frac{1}{T-T_0}\sum_{t=T_0+1}^{T}\eta^2 \tilde{\tau}_{j, t}^2.
\]
By definition, $p_{\pb}\ge 1/N$. 

When $N$ is odd, we construct $(G, \tau)$ such that the first $N-1$ units follow the above prescription and the $N$-th unit can be arbitrary. Then $p_{\pb}\ge 1/N$ whenever $I \neq N$. Since $\Prob(I = N) = 1/N$, we conclude that $\Prob_{G, H_{1,\eta}}(p_{\nLTO} \ge 1/N) \le 1/N$ for any $\eta > 0$.

\end{proof}

\subsection{Minor lemmas}

\begin{lemma}
\label{lemma:binomial_coef_identity}
For any integers $N > r > 0$, the sequence
\[
s \mapsto A_{N, r}(s) := \frac{1}{s}\left[\binom{N}{r+1} - \binom{N-s}{r+1}\right]
\]
is decreasing.
\end{lemma}
\begin{proof}
    Observe the two following identities on binomial coefficients: 
\begin{gather*}
    \binom{N - s}{r + 1} = \sum_{k=r}^{N -1-s} \binom{k}{r} \\
    s\binom{N - 1 -s}{r} \leq \sum_{k' = N - s}^{N-1} \binom{k'}{r}.
\end{gather*}
The first is true because of the hockey-stick identity, and the second is by monotonicity. Summing these inequalities, we find
\[
\binom{N - s}{r + 1} + s\binom{N - 1 -s}{r} \leq \sum_{k = r}^{N-1} \binom{k}{r} = \binom{N}{r+1},
\]
which is equivalent to
\[
\binom{N - s}{r + 1} + s\left[\binom{N - s}{r+1} - \binom{N-1-s}{r+1} \right] \leq (s+1)\binom{N}{r+1} - s\binom{N}{r+1}.
\]
Rearranging, $A_{N, r}(s) \geq  A_{N, r}(s+1)$.
\end{proof}

\end{document}